\documentclass[journal]{IEEEtran}

\usepackage{amsfonts}
\usepackage{amssymb}
\usepackage{stfloats}
\usepackage{graphicx}
\usepackage{amsmath}
\usepackage{array}
\usepackage{epstopdf}
\usepackage{amsthm}
\usepackage{algorithm}
\usepackage[noend]{algpseudocode}
\usepackage{amsmath}
\usepackage{graphics}
\usepackage{epsfig}

\usepackage{bm}
\usepackage{upgreek}
\usepackage{caption}
\usepackage{subcaption}
\usepackage{tabularx}
\usepackage{cite}
\newlength\myindent
\usepackage[usenames]{color}

\newcommand{\trans}{^{\mathsf{T}}}

\DeclareMathOperator{\argmin}{argmin}

\usepackage{mathtools}

\newtheorem{Proposition}{Proposition}
\newtheorem{Remark}{Remark}

\begin{document}

\title{Accelerating Federated Edge Learning via \\Topology Optimization}

\author{
    Shanfeng Huang,~\IEEEmembership{Graduate Student Member,~IEEE}, Zezhong Zhang,~\IEEEmembership{Member,~IEEE}, Shuai Wang,~\IEEEmembership{Member,~IEEE}, Rui Wang,~\IEEEmembership{Member,~IEEE},  and Kaibin Huang,~\IEEEmembership{Fellow,~IEEE}

    \thanks{	
    S. Huang is with the Department of Electrical and Electronic Engineering, The University of Hong Kong, Hong Kong, and also with the Department of Electronic and Electrical Engineering, Southern University of Science and Technology, Shenzhen 518055, China (e-mail: sfhuang@eee.hku.hk).
    
    Zezhong Zhang is with The Future Network of Intelligence Institute (FNii), The Chinese University of Hong Kong (Shenzhen), Shenzhen 518172, China (email: zhangzezhong@cuhk.edu.cn).

    Shuai Wang is with the Shenzhen Institute of Advanced Technology (SIAT), Chinese Academy of Sciences, Shenzhen 518055, China (e-mail: s.wang@siat.ac.cn).
	
    R. Wang is with the Department of Electrical and Electronic Engineering, Southern University of Science and Technology, Shenzhen, China, and also with the Research Center of Networks and Communications, Peng Cheng Laboratory, Shenzhen, China (e-mail: wang.r@sustech.edu.cn).

	K. Huang is with the Department of Electrical and Electronic Engineering, The University of Hong Kong, Hong Kong (e-mail: huangkb@eee.hku.hk).

    {Corresponding author: Z. Zhang and R. Wang.}
}
}

\maketitle

\begin{abstract}
    Federated edge learning (FEEL) is envisioned as a promising paradigm to achieve privacy-preserving distributed learning. However, it consumes excessive learning time due to the existence of straggler devices. In this paper, a novel topology-optimized federated edge learning (TOFEL) scheme is proposed to tackle the heterogeneity issue in federated learning and to improve the communication-and-computation efficiency. Specifically, a problem of jointly optimizing the aggregation topology and computing speed is formulated to minimize the weighted summation of energy consumption and latency. To solve the mixed-integer nonlinear problem, we propose a novel solution method of penalty-based successive convex approximation, which converges to a stationary point of the primal problem under mild conditions. To facilitate real-time decision making, an imitation-learning based method is developed, where deep neural networks (DNNs) are trained offline to mimic the penalty-based method, and the trained imitation DNNs are  deployed at the edge devices for online inference. Thereby, an efficient  imitate-learning based approach is seamlessly integrated into the TOFEL framework. Simulation results demonstrate that the proposed TOFEL scheme accelerates the federated learning process, and achieves a higher energy efficiency. Moreover, we apply the scheme to 3D object detection with multi-vehicle point cloud datasets in the CARLA simulator. The results confirm the superior learning performance of the TOFEL scheme over conventional designs with the same resource and deadline constraints.
\end{abstract}


\section{Introduction}
Recent years have witnessed unprecedented successes of deep-learning based artificial intelligence (AI) in a wide range of applications, such as speech recognition, image classification, autonomous driving. On the other hand, the massive Internet of Things (IoT) devices and mobile terminals generate a vast amount of data that can be employed for AI model training. However, sending these massive data to the central servers causes concerns on privacy. To address the issue, federated learning has emerged as a promising paradigm to achieve privacy-preserving distributed learning such that the original datasets are kept in their generated devices and only the neural network model parameters are shared \cite{jakub2016FL,CX,GuiGuan1}.

\subsection{Federated Edge Learning}
Although federated learning was originally proposed for the systems with wired connections, many intelligent systems are implemented with wireless links, such as IoT smart surveillance and vehicle-to-everything (V2X) autonomous driving. This results in a new research area called federated edge learning (FEEL) which concerns the implementation of federated learning in the wireless networks \cite{Zhu2019BroadbandAggr,chen2021joint,Du2020SGquant,guo2021analoggrad,amiri2020edgelearn}.
In a FEEL framework, each round of the iterative learning process involves the broadcasting of a global model to edge devices, the uploading of local gradients calculated from the locally stored datasets at the edge devices, as well as the aggregation of the local gradients and global model update at the edge server.

The uploading of the high-dimensional local gradients incurs excessive communication loads, resulting in high communication latency. Since the model is usually trained for subsequent edge inference tasks \cite{GuiGuan2}, and outdated model will lead to low inference accuracy, communication latency becomes a crucial issue in FEEL. To reduce the communication latency and accelerate FEEL, a vein of active research is devoted to designing communication-efficient FEEL exploiting the sparsity of gradient updates \cite{yujun2018deepgrad, scattler2019sparsity} and low-resolution gradient/model  quantization \cite{alistarh2017QSGD,jeremy2018signSGD}. Moreover, one-bit gradient quantization for FEEL incorporating wireless channel hostilities is investigated in \cite{zhu2021onebit}, where model convergence is demonstrated in the presence of channel noise. In \cite{Du2020SGquant}, the authors further improve the compression ratio by a novel hierarchical vector quantization scheme using low-dimensional Grassmannian codebooks.

Another challenge associated with FEEL is the heterogeneity of wireless channels and the resources of edge devices. Due to channel fading, the connections between certain edge devices and the edge server may suffer from deep fading in some iteration rounds. In this case, the communication latency for those devices becomes overwhelming even with highly compressed gradient/model. Moreover, the different computation capabilities and different number of local training samples will result in different computation time. Such heterogeneity will cause significant slow-down in the runtime due to the existence of stragglers and in turn exacerbate the convergence speed of the FEEL system. Several strategies have been proposed to deal with heterogeneity in federated learning. The first is to employ asynchronous update that allows model aggregation without waiting for slow-responding devices \cite{xie2019asynchronous,sprague2019async,chen2018lag,chen2019asynchronous}. While asynchronous training have proven to be faster than their synchronous counterparts due to their straggler resistance, they often result in convergence to poorer results \cite{jianmin2016revisit}, and the dynamics in asynchronous FEEL bring more challenges in parameter tuning \cite{Chen2019EfficientAR}. Another more radical strategy is to directly discard the slow devices by designing various client selection schemes \cite{xu2021client,xia2020client,Mohammed2021client}. Although discarding the slow devices can reduce the latency for one-round iteration, the number of iterations required generally becomes larger since such schemes cannot make full use of the valuable data resources at the slower devices.

Recently, a line of research works have made some attempts to deal with the communication efficiency and heterogeneity issue from the perspective of aggregation topology design. The authors in \cite{abad2020hierarchical} proposed a hierarchical FEEL framework, where small base stations were introduced to orchestrate FEEL among the devices
within their cells, and periodically exchange model updates
with the macro base station for global consensus. It was shown that the hierarchical FEEL scheme significantly reduces the communication latency without sacrificing the accuracy. Similarly, the authors in \cite{Luo2020HFL} investigated joint computation and communication resource allocation in a device-edge-cloud hierarchical FEEL, and designed a device-edge association scheme to address the heterogeneity of wireless channels and device resources. Moreover, to involve more participates in the model training process, the authors in \cite{Metropolis,CMZ} propose to let devices with deep channel fading send local model updates to neighboring devices with higher channel gains, which is called the lazy Metropolis update approach. Nevertheless, most of the existing hierarchical FEEL frameworks are developed with fixed topology. The design and optimization of more flexible aggregation topology for FEEL systems remains largely uncharted.

\subsection{Challenges and Contributions}
In this paper, we investigate hierarchical FEEL with an adjustable gradient uploading and aggregation topology. By exploiting device-to-device (D2D) communications, any device in the FEEL system can act as an aggregation-and-forward (AF) device that can receive the gradients from other devices, and aggregates its own gradient with the received ones. The aggregated gradient is then sent to the edge server or another AF device for further aggregation. This scheme is referred to as topology-optimized federated edge learning (TOFEL) in this paper. The
advantages of TOFEL can be intuitively explained as follows. The device with harsh channel condition to the edge server can flexibly choose a nearby device with favorable channel to aggregate and forward its gradient, which avoids the high transmission latency to the edge server. In addition to the link power gain, the gradient aggregation at the AF devices will not increase their uplink traffic loads, compared with conventional relaying communications. Moreover, the computing speed adaptation via dynamic voltage and frequency scaling (DVFS) \cite{Singh2019DVFS} can be exploited jointly with aggregation topology optimization to boost the computation efficiency of the devices.

However, the design of a TOFEL system is nontrivial. In this paper, we make some attempts to answer to following two relevant questions:
\begin{enumerate}
\item \textit{How to design the optimal topology for gradient aggregation according to dynamic wireless channels and heterogeneous resources at the edge devices?}
\item \textit{How to adjust the computing speed of the devices to adapt to the TOFEL framework to improve computation efficiency?}
\end{enumerate}

To answer these two questions, we formulate the problem of joint topology and computing speed optimization as a mixed-integer nonlinear program (MINLP), which aims at minimizing the weighted summation of energy consumption and latency. Conventional combinatorial optimization methods to solve MINLP, e.g., branch-and-bound searching, suffer from high complexity when the searching space is large. Thus, they can hardly meet the real-time constraints of wireless scheduling. Our contributions in solving the above MINLP are summarized below:
\begin{itemize}
    \item \textbf{Efficient penalty-based method to solve MINLP.} A penalty-based successive convex approximation (SCA) method is developed to transform the MINLP into an equivalent continuous optimization problem that can be solved efficiently. Also, the optimization results demonstrate that our proposed TOFEL scheme can remarkably accelerate the federated learning process, and achieve higher energy efficiency in the meanwhile.
    \item \textbf{Imitation-learning based implementation for real-time decision making.} A deep neural network (DNN) is designed to imitate the optimization algorithm. Specifically, we use the penalty-based method to generate massive demonstrations to train the DNN in an offline manner. Once the imitation learning DNN is trained, it can be efficiently deployed at the edge devices to make distributed inference in an online manner, and the inference process is seamlessly integrated into the TOFEL framework.
\end{itemize}

Finally, we implement our proposed TOFEL scheme for 3D object detection task with multi-vehicle point cloud datasets in CARLA simulator. It is shown that with the same resource and deadline constraints, the proposed TOFEL scheme can achieve much higher detection accuracy than that of existing FEEL schemes.

\section{System Model}
We consider a federated edge learning (FEEL) system as shown in Fig. \ref{fig:TopoFL}, where a set of $K$ edge devices  denoted as $\mathcal K=\{1,\cdots, K\}$ are performing a federated learning task with the help of an edge server (indexed by $K+1$). The federated learning process consists of iterative local gradient computing, uploading and global model updating. Each iteration of FEEL is called a communication round, or round for short. Data communications among the edge devices and edge server are via wireless links. It is assumed that the server has perfect knowledge of wireless channel gains and edge devices' computation characteristics, which can be obtained by feedback. Using this information, the edge server determines the communication and computation parameters at the beginning of each round, i.e., the gradient aggregation topology of the FEEL system, and the computing speed of the edge devices.

\subsection{Federated Learning Model}

\begin{figure}[t]
    \centering
    \includegraphics[width=\linewidth]{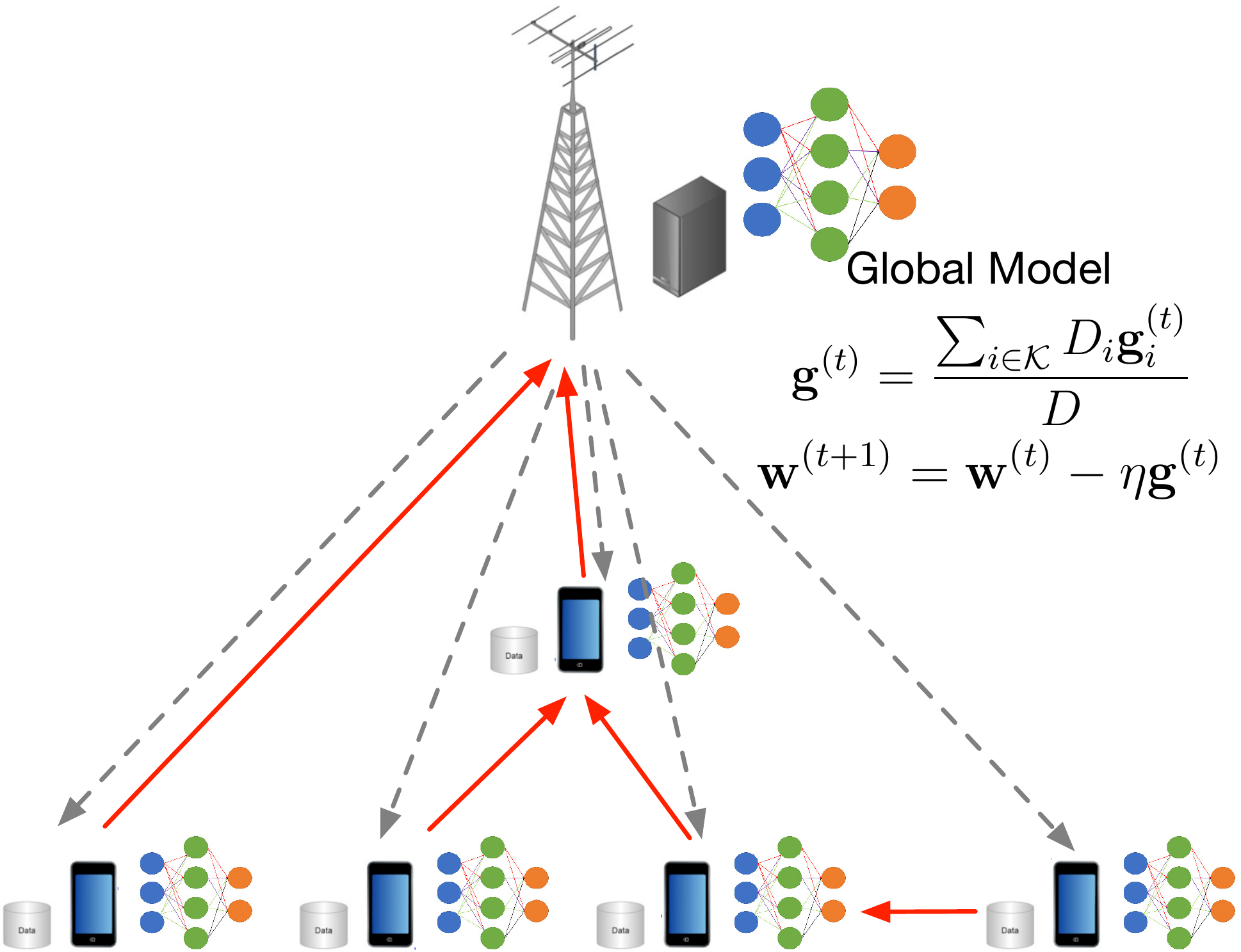}
    \caption{TOFEL: Topology-optimized federated edge learning framework.}
    \label{fig:TopoFL}
\end{figure}
As shown in Fig. \ref{fig:TopoFL}, a global model, represented by the parameter set $\mathbf w$, is trained collaboratively by the edge devices and the edge server. Each device (say device $i$) has a local dataset $\mathcal D_i$ with $D_i=|\mathcal D_i|$ sample points. The local loss function of the model parameter $\mathbf w$ on the dataset $\mathcal D_i$ is given by
\begin{align}
    \text{(Local loss function)}\quad
    F_i(\mathbf w)=\frac{1}{D_i} \sum_{(\mathbf x_n, y_n)\in\mathcal D_i}\ell(\mathbf w; \mathbf x_n, y_n),
\end{align}
where $\ell(\mathbf w; \mathbf x_n, y_n)$ is the sample-wise loss function measuring the prediction error of the model $\mathbf w$ on the training sample $\mathbf x_n$ with respect to its ground-true label $y_n$. Then the global loss function on all the distributed datasets can be expressed as
\begin{align}
    \text{(Global loss function)}\,
    F(\mathbf w)&=\frac{\sum_{(\mathbf x_n, y_n)\in\cup_i\mathcal D_i}\ell(\mathbf w; \mathbf x_n, y_n)}{|\cup_i\mathcal D_i|}\nonumber\\
    &=\frac{\sum_{i\in\mathcal K}D_iF_i(\mathbf w)}{D},
\end{align}
where $D=\sum_{k=1}^{K}D_i$ is the total number of data samples from all the edge devices.

The learning process is to minimize the global loss function $F(\mathbf w)$, namely,
\begin{align}\label{prob:fl}
    \mathbf w^*=\argmin F(\mathbf w).
\end{align}
$F(\mathbf w)$ can be computed directly by uploading all the local data to the edge server \cite{huang2021RISEL}, which unfortunately will cause privacy concern. To this end, FEEL solves problem (\ref{prob:fl}) in a distributed manner without transmitting the datasets to the edge server. For ease of elaboration, we focus on gradient-averaging implementation in the subsequent exposition, while similar design also applies to model-averaging implementation.

In each round, say $t$-th round, the edge server broadcasts the current global model $\mathbf w^{(t)}$ to all the edge devices. After receiving the model $\mathbf w^{(t)}$, each edge device calculates the gradient with its local dataset as:
\begin{align}
    \text{(Local gradient)}\quad
    \mathbf g_i^{(t)}=\nabla F_i(\mathbf w^{(t)}).
\end{align}
Then, the local gradients are uploaded to the edge server via wireless communication for gradient aggregation:
\begin{align} \label{eq:grad}
    \text{(Gradient aggregation)}\quad
    \mathbf g^{(t)}=\frac{\sum_{i\in\mathcal K}D_i\mathbf g_i^{(t)}}{D}.
\end{align}
After that, the global model is updated by stochastic gradient descent (SGD) as
\begin{align}
    \text{(Global model update)}\quad
    \mathbf w^{(t+1)}=\mathbf w^{(t)}-\eta\mathbf g^{(t)},
\end{align}
where $\eta$ is the learning rate. The above process iterates until the model converges.

Note that the system also supports model parameter aggregation. 
In such a case, the local gradient calculation step in equation (4) becomes the local parameter update step, which is given by
\begin{align}
    \text{(Local parameter update)}\quad
    \mathbf u^{(t)}_i=\mathbf w^{(t)}-\eta\nabla F_i(\mathbf w^{(t)}).
\end{align}
Then, the local parameters are uploaded to the edge server via wireless communication for parameter aggregation:
\begin{align}
    \text{(Global model update)}\quad
    \mathbf w^{(t+1)}=\frac{\sum_{i\in\mathcal K}D_i\mathbf u^{(t)}_i}{D}.
\end{align}
It can be seen that the gradient aggregation and model parameter aggregation are mathematically equivalent.

\subsection{Gradient Aggregation Topology}
In most of the existing literature on FEEL, the edge devices upload the local gradients directly to the edge server \cite{Zhu2019BroadbandAggr,chen2021joint,amiri2020edgelearn,Du2020SGquant,guo2021analoggrad}, or via fixed relays \cite{abad2020hierarchical,Luo2020HFL}.
Such FEEL systems suffers from the ``barrel effect": the latency for one-round iteration is determined by the device with longest computation and transmission time. For instance, if the wireless channels between certain devices and the edge server are very poor, the long uploading time of the local gradients of these devices will drag the latency of the one-round iteration significantly. Meanwhile, to meet the deadline requirement, more power may be consumed by them to shorten the computation time of their local updates. To address this issue, we introduce a novel TOFEL framework as shown in Fig. \ref{fig:TopoFL}, where the devices with poor wireless channels to the edge server may send their local gradients to their nearby devices. The receiving devices aggregate their local gradients with the received gradients, and forward the results to one another device or the edge server for further aggregation.

For ease of notation, the edge devices and the edge server are collectively referred to as ``nodes'' in the subsequent exposition. The nodes that aggregate and forward gradients are called AF nodes. The set of all nodes is denoted as $\mathcal K^0=\mathcal K\cup\{K+1\}$ with $K+1$ representing the index of the edge server. Denote $I_{i,j}$ as an indicator which equals 1 if node $i$ transmits its gradient to node $j$, and 0 otherwise. Thus, $\mathbf I=\{I_{i,j}|i\in \mathcal K, j\in\mathcal K^0\}$ uniquely specifies the gradient transmission and aggregation topology in the federated learning system. If node $j$ works as an AF node, it will aggregate its own gradient vector with all its received gradient vectors. Thus, the aggregated gradient vector at node $j$ is given by
\begin{align}\label{eq:gradaggr}
\tilde{\mathbf g}_j=\frac{\sum_{i\in\mathcal K}I_{i,j}\tilde D_i\tilde{\mathbf g}_i+D_j\mathbf g_j}{\sum_{i\in\mathcal K}I_{i,j}\tilde D_i+D_j},
\end{align}
where $\tilde D_i=\sum_{k\in\mathcal K}I_{k,i}\tilde D_k+D_i$ is the accumulated number of data samples at node $i$.
The aggregated gradient vector $\tilde{\mathbf g}_j$ is then transmitted to one another nodes for further aggregation.
\begin{Proposition}\label{prop:equiv}
    With the local gradient aggregation scheme in (\ref{eq:gradaggr}), if all the $K+1$ nodes in the FEEL system form a tree aggregation topology, the final aggregated gradient vector at the root node equals the aggregated gradient in (\ref{eq:grad}).
\end{Proposition}
\begin{proof}
    Please refer to Appendix A.
\end{proof}
From Proposition \ref{prop:equiv}, we have the following constraints on the gradient aggregation topology: 1) there should be no ring in the topology, and hence an arbitrary node can only transmit its gradient to one AF node; 2) the edge server should be the root node of the aggregation tree.

\begin{Remark}
    Note that in the TOFEL framework, each edge device (say the node $j$) acts as a dual functional relay and source node, which forwards the updated gradient $\tilde{\mathbf g}_j$ (rather than the received gradients $\{\mathbf g_{i}|I_{i,j}=1\}$) to its father node. Hence, the TOFEL framework will not increase the uploading load of each edge device. This is fundamentally different  from conventional multi-hop relay systems where the relay only forwards information from other nodes.
\end{Remark}

Finally, the convergence of TOFEL is proved in the following Proposition. 

\begin{Proposition}\label{prop:convergence}
Assume that the function $F_i(\mathbf{w})$ is $\mu$-strongly convex and is twice differentiable with $\nabla^2_{\mathbf{w}}F_i(\mathbf{w})\preceq L\mathbf{I}$.
Then with $\eta=\frac{1}{L}$, the TOEFL framework satisfies
\begin{align}
F(\mathbf{w}^{(t+1)})-F(\mathbf{w}^{*})
&\leq
\left(1-\frac{\mu}{L}\right)^{i+1}
\left[F(\mathbf{w}^{(o)})-F(\mathbf{w}^{*})\right]. \nonumber
\end{align}
\end{Proposition}

\begin{proof}
See \cite{wang2021edge} Appendix A.
\end{proof}


\subsection{Local Computation Model}
Following the computation model in \cite{zhang2018shufflenet, Huang2020MEC}, we denote $N_{\mathsf{FLOP}}$ as the number of floating point operations (FLOPs) needed for processing each training data sample. Furthermore, we define $f_i^{c}$ (in cycles/s) as the clock frequency of the CPU at the $i$-th edge device. It follows that the computing speed of the $i$-th edge device is given by $f_i=f_i^{c}\times z_i$, with $z_i$ denoting the number of FLOPs per CPU cycle.  Thus, the computation time for local gradient calculation at the $i$-th edge device is given by
\begin{align}
    t_i^{\text{cmp}}(f_i)=\frac{D_iN_{\mathsf{FLOP}}}{f_i}.
\end{align}
For a CMOS circuit, the power consumption of a processor can be modeled as a function of the clock frequency:
\begin{align}
    P_i^{\text{cmp}}(f_i)=\kappa_i^{c}(f_i^{c})^3=\kappa_if_i^3,
\end{align}
where $\kappa_i=\kappa_i^{c}/z_i^3$ and $\kappa_i^{c}$ is depending on the chip architecture \cite{Burd1996CPU}.
Moreover, the energy consumption for local model update at the $i$-th edge device can be obtained as
\begin{align}
    E_i^{\text{cmp}}(f_i)=\kappa_i f_i^3t_i^{\text{cmp}}=\kappa_i D_iN_{\mathsf{FLOP}}f_i^2.
\end{align}
The computing speed $f_i$ can be adjusted between $f_{\min}$ and $f_{\max}$ by DVFS, so does the computing energy consumption. Moreover, the computational complexity for aggregating the gradient vectors is negligible compared with that of local gradient computation. Hence, the computation time and energy consumption for gradient aggregation are neglected in this work.

\subsection{Wireless Communication Model}
In our considered system, the nodes can cooperate with each other to aggregate and forward their local gradients via wireless communication links. Specifically, the $i$-th node can either transmit its local gradient $\mathbf g_i$ to the edge server directly, or to the $j$-th node where  $j\in\mathcal K\backslash\{i\}$ via device-to-device (D2D) communication. Denote the distance between node $i$ and $j$ as $d_{i,j}$, the pathloss coefficient as $\alpha$, and the small-scale fading as $h_{i,j}\sim\mathcal{CN}(0,1)$. The channel gain from node $i$ to node $j$ is given by $H_{i,j}=g_0(d_0/d_{i,j})^{\alpha}|h_{i,j}|$, where $g_0$ is the pathloss constant and $d_0$ is the reference distance.  Orthogonal frequency bands are allocated to all the nodes for gradient transmission. Let $w_i$ be the bandwidth allocated to node $i\in\mathcal K$. Hence, the achievable data transmission rate from node $i$ to node $j$ is given by
\begin{align}
    r_{i,j}=w_i\log_2\left(1+\frac{PH_{i,j}}{\sigma^2}\right),
\end{align}
where $P$ and $\sigma^2$ are the transmit power and channel noise at each node, respectively.
As a result, the transmission time and energy consumption for edge device $i$ is given by
\begin{align}
    t_{i}^{\text{cmm}}=\frac{B}{\sum_{j\in\mathcal K^0}I_{i,j}r_{i,j}},
\end{align}
\begin{align}
    E_{i}^{\text{cmm}}=\frac{BP}{\sum_{j\in\mathcal K^0}I_{i,j}r_{i,j}},
\end{align}
where $B$ is the number of bits for transmitting a gradient vector, and $\sum_{j\in\mathcal K^0}I_{i,j}=1$ since node $i$ can only transmit its gradient vector to one node $j$.

\section{Problem Formulation}
In this paper, we aim at accelerating the federated learning process, as well as reducing the energy consumption for training and transmission. Specifically, our objective is to minimize the weighted summation of energy consumption and latency for each round of the TOFEL framework. To achieve this goal, we jointly design the topology matrix $\mathbf I=\{I_{i,j}|i\in \mathcal K, j\in\mathcal K^0\}$ and optimize the local computing speed of each edge device $\mathbf f=\{f_i|i\in\mathcal K\}$ via DVFS. In the following, we first introduce the constraints on communication, computation and topology design, and then the overall scheduling problem.

\begin{figure}[t]
    \centering
    \includegraphics[width=\linewidth]{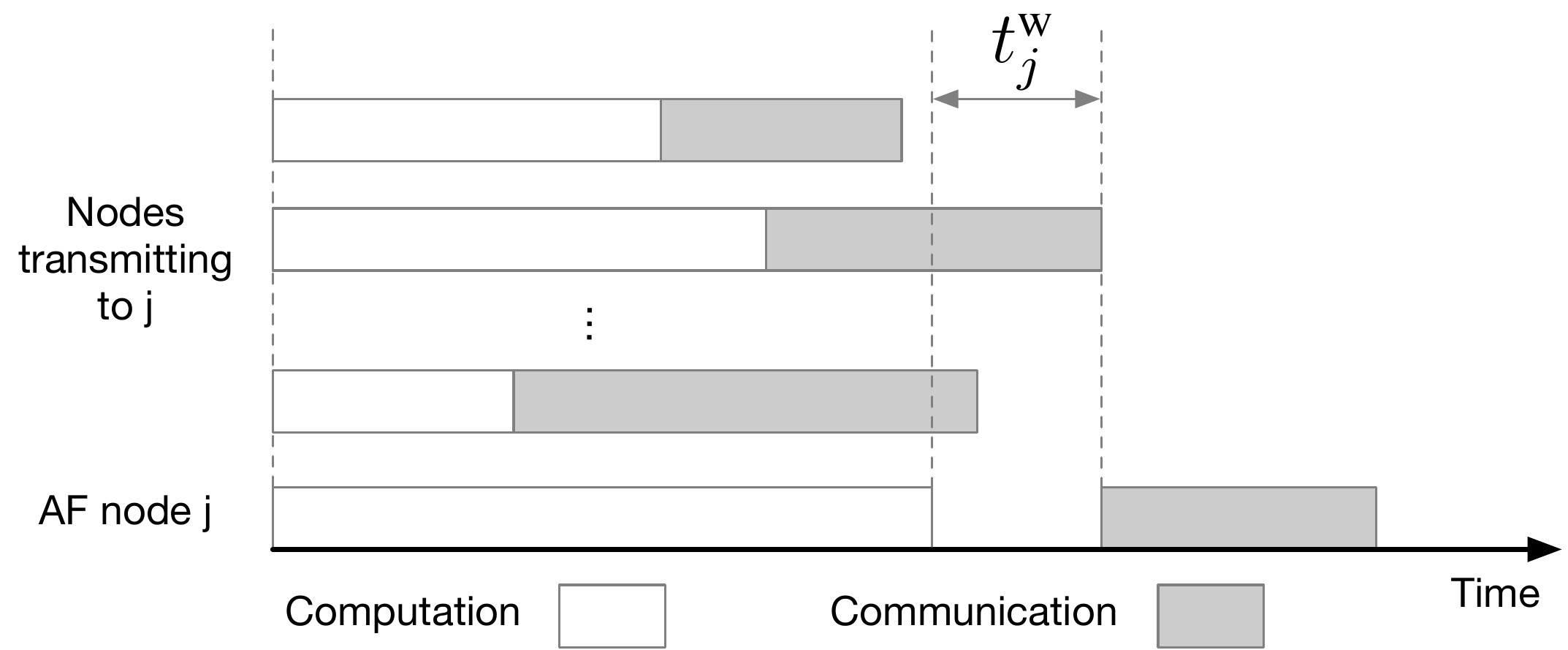}
    \caption{An illustration of the computation and communication time constraint.}
    \label{fig:waittime}
\end{figure}
Since an AF node (say the node $j$) should receive and aggregate gradient vectors from its neighboring nodes,  it can only start transmitting after it has received all the gradient vectors from nodes $i$ with $I_{i,j}=1$. Hence, if the local computation time of node $j$ is shorter than the local computation time plus transmission time of nodes $i$ with $I_{i,j}=1$, it has to wait for a period of time $t_j^{\text{w}}\geq 0$, as shown in Fig. \ref{fig:waittime}. Therefore, we have the following constraints on the computation and communication time of all the connected node pairs:
\begin{align}
    \left(\frac{D_iN_{\mathsf{FLOP}}}{f_i}+\frac{B}{\sum_{j\in\mathcal K^0}I_{i,j}r_{i,j}}\right)I_{i,j}\leq \frac{D_{j}N_{\mathsf{FLOP}}}{f_{j}}+t_j^{\text{w}},\nonumber\\i\in\mathcal K,\,j\in\mathcal K.
\end{align}
It is easy to observe that the larger $t_j^{\text{w}}$ is, the faster computing speed $f_j$ is required for the AF node $j$, leading to larger computation energy consumption of node $j$. Hence, the optimal $t_j^{\text{w}}$ must be 0. In this case, if the local computing time of node $j$ is shorter than the local computing time plus transmission time of nodes $i$ with $I_{i,j}=1$, we can lower its computing speed to reserve energy consumption. Thus, the refined computation and communication time constraints are given below:
\begin{align} \label{eq:CCtimeconst}
    &\text{(Comp. \& comm. time const.)}\nonumber\\
    &\left(\!\frac{D_iN_{\mathsf{FLOP}}}{f_i}\!+\!\frac{B}{\sum_{j\in\mathcal K^0}I_{i,j}r_{i,j}}\!\right)\!I_{i,j}\!\leq\! \frac{D_{j}N_{\mathsf{FLOP}}}{f_{j}},\,i\!\in\!\mathcal K,j\!\in\!\mathcal K.
\end{align}

Moreover, we assume that each node $i$ can only transmit its local model to one node, and at least one node is directly connected to the edge server. Hence, we have the following constraint on the node topology:
\begin{align}
    \text{(Node topology constraints)}\,
    &\sum_{j\in\mathcal K^0}I_{i,j}=1, I_{i,i}=0,\,\forall i\in\mathcal K. \label{eq:node_topo_1}\\
    &\sum_{i\in\mathcal K}I_{i,K+1}\geq 1.\label{eq:node_topo_2}
\end{align}

\begin{Proposition}\label{prop:noring}
    Equation (\ref{eq:CCtimeconst}) guarantees that there is no ring in the topology of TOFEL framework. Together with equations (\ref{eq:node_topo_1}) and (\ref{eq:node_topo_2}), the topology of the federated learning system is defined as a tree topology, which guarantees that all the local gradient vectors are aggregated once and only once.
\end{Proposition}
\begin{proof}
    Please refer to Appendix B.
\end{proof}

With a tree topology, the total latency of one-round TOFEL is determined by the largest computation and communication time of the nodes that are connected to the edge server, and therefore is given by
\begin{align} \label{eq:latency}
    t^{\text{tot}}&=\max_{i\in\mathcal I_{\text{edge}}}\left\{\frac{D_iN_{\mathsf{FLOP}}}{f_i}+\frac{B}{\sum_{j\in\mathcal K^0}I_{i,j}r_{i,j}}\right\}\nonumber\\
    &=\max_{i\in\mathcal K}\left\{\frac{D_iN_{\mathsf{FLOP}}}{f_i}+\frac{B}{\sum_{j\in\mathcal K^0}I_{i,j}r_{i,j}}\right\},
\end{align}
where $\mathcal I_{\text{edge}}=\{i|I_{i,K+1}=1,i\in\mathcal K\}$ denotes the set of nodes that are connected to the edge server. Moreover, the total computation and communication energy consumption for one-round iteration is given by
\begin{align} \label{eq:energy}
    E^{\text{tot}}=\sum_{i\in \mathcal K}\left(\kappa_i D_iN_{\mathsf{FLOP}}f_i^2+\frac{BP}{\sum_{j\in\mathcal K^0}I_{i,j}r_{i,j}}\right).
\end{align}

Therefore, the optimization problem which minimizes the weighted summation of energy consumption and latency with computation, communication, and topology constraints can be formulated as
\begin{align}
    \textbf{(P1)}\,
    \min_{\mathbf I,\mathbf f} \, &\sum_{i\in \mathcal K}\left(\kappa_i D_iN_{\mathsf{FLOP}}f_i^2+\frac{BP}{\sum_{j\in\mathcal K^0}I_{i,j}r_{i,j}}\right)\nonumber\\
    &+\mu\max_{i\in\mathcal K}\left\{\frac{D_iN_{\mathsf{FLOP}}}{f_i}+\frac{B}{\sum_{j\in\mathcal K^0}I_{i,j}r_{i,j}}\right\}\\
    \mathrm{s.t.} \,
    &\left(\frac{D_iN_{\mathsf{FLOP}}}{f_i}+\frac{B}{\sum_{j\in\mathcal K^0}I_{i,j}r_{i,j}}\right)I_{i,j}\leq \frac{D_jN_{\mathsf{FLOP}}}{f_j},\nonumber\\
    &\qquad\qquad\qquad\qquad\qquad\qquad i\in\mathcal K,j\in\mathcal K,\nonumber\\
    &\sum_{j\in\mathcal K^0}I_{i,j}=1, I_{i,i}=0,\,i\in\mathcal K,\nonumber\\
    &\sum_{i\in\mathcal K}I_{i,K+1}\geq 1,\nonumber\\
    &I_{i,j}\in\{0,1\},\,i\in\mathcal K, j\in\mathcal K^0,\label{eq:zerooneconst}\\
    &f_{\min}\leq f_i\leq f_{\max},\, i\in\mathcal K,\label{eq:freqconst}
\end{align}
where $\mu$ is the weighting factor to balance energy consumption and latency.

In the following, we first introduce a penalty-based solution to ($\mathbf{P1}$) where an SCA-based algorithm is used. An imitation learning based method is then introduced to train DNNs based on datasets formed by the penalty-based solution.
Note that problem ($\mathbf{P1}$) can be generalized to joint optimization across multiple federated iterations by minimizing the weighted summation of energy consumption and latency across multiple federated iterations. For example, we can adjust the transmit powers at different rounds to accelerate the convergence speed. In such a case, the original power constant $P$ would become a vector variable $\mathbf{p}=[p^{(1)},\cdots,p^{(T)}]^T$ with constraint $\frac{1}{T}\sum_tp^{(t)}\leq P$. The variables to optimize become the topology $\mathbf I^{(t)}$, frequency $\mathbf f^{(t)}$ and transmit power $\mathbf p^{(t)}$ in each round $t$.


\section{Penalty-Based Solution}
Problem ($\mathbf{P1}$) is an MINLP problem. Direct branch-and-bound approach could solve such combinatorial problems. However, prohibitive  complexity may incur in the case of a large number of devices due to the high-dimensional search space. Therefore, we resort to a penalty-based method to transform the MINLP problem into a continuous problem.

\subsection{Penalty-Based Continuous Reformulation}
To tackle the discontinuity , we first relax the binary constraint $I_{i,j}\in\{0,1\}$ into a linear constraint $I_{i,j}\in [0,1]$, $i\in\mathcal K, j\in \mathcal K^0$. However, in general, the relaxation is not tight, i.e., the solution to the relaxed problem could be $0<I_{i,j}<1$, $\forall i,j$. Therefore, to promote a binary solution for the relaxed variable $I_{i,j}$, we regularize the objective function with a penalty term as in \cite{Rinaldi2009zero-one,Bach2012Sparsity}, which only depends on the relaxed variable $\mathbf I$. Accordingly, the approximate reformulation with regularized penalty term of ($\mathbf {P1}$) is given by
\begin{align}
    \textbf{(P2)}\,\,
    \min_{\mathbf I,\mathbf f} \, &\sum_{i\in \mathcal K}\left(\kappa_i D_iN_{\mathsf{FLOP}}f_i^2+\frac{BP}{\sum_{j\in\mathcal K^0}I_{i,j}r_{i,j}}\right)\nonumber\\
    &+\mu\max_{i\in\mathcal K}\left\{\frac{D_iN_{\mathsf{FLOP}}}{f_i}+\frac{B}{\sum_{j\in\mathcal K^0}I_{i,j}r_{i,j}}\right\}+\varphi(\mathbf I)\\
    \mathrm{s.t.} \quad
    & 0\leq I_{i,j}\leq 1, \, i\in\mathcal K, j\in\mathcal K^0,\label{eq:continuousI}\\
    &(\ref{eq:CCtimeconst}), (\ref{eq:node_topo_1}), (\ref{eq:node_topo_2}), \text{ and } (\ref{eq:freqconst}).\nonumber
\end{align}
where $\varphi(\mathbf I)$ is a penalty function to penalize the violation of the zero-one integer constraints. A celebrated penalty function was introduced in \cite{Giannessi1976}, where the penalty function is set as
\begin{align} \label{eq:penaltyterm}
    \varphi(\mathbf I)=\frac{1}{\beta}\sum_{i\in \mathcal K}\sum_{j\in\mathcal K^0}I_{i,j}(1-I_{i,j}),
\end{align}
where $\beta>0$ is the penalty parameter.
According to Proposition 1 in \cite{lucidi2010exact}, with penalty term (\ref{eq:penaltyterm}), there exist a value $\bar\beta>0$ such that, for any $\beta\in[0,\bar\beta]$, problem ($\mathbf {P1}$) and ($\mathbf{P2}$) have the same minimum points. That is, ($\mathbf {P1}$) and ($\mathbf{P2}$) are equivalent with a proper choice of $\beta$.

\subsection{ SCA-Based Algorithm}

To begin with, we need to derive an equivalent difference-of-convex (DC) formulation for problem ($\mathbf{P2}$). Among the objective function and all the constraints, the only non-DC part is constraint (\ref{eq:CCtimeconst}). To derive its DC reformulation, we introduce a  slack variable $\mathbf T=\{t_{i,j}|i\in\mathcal K, j\in\mathcal K\}$ which satisfies $\frac{I_{i,j}}{f_i}\leq t_{i,j}^2$ and a slack variable $\mathbf S=\{s_{i,j}|i\in\mathcal K,j\in\mathcal K\}$ which satisfies $\frac{I_{i,j}}{\sum_{j\in\mathcal K^0}I_{i,j}r_{i,j}}\leq s_{i,j}^2$ to transform the bilinear terms in (\ref{eq:CCtimeconst}) into quadratic terms. Consequently, problem ($\mathbf {P2}$) can be equivalently written as the following form:
\begin{align}
    \textbf{(P3)}\nonumber\\
    \min_{\mathbf I,\mathbf f}  &\sum_{i\in \mathcal K}\!\!\left(\!\!\kappa_i D_iN_{\mathsf{FLOP}}f_i^2\!+\!\frac{BP}{\sum_{j\in\mathcal K^0}\!r_{i,j}I_{i,j}}\!+\!\frac{1}{\beta}\!\!\sum_{j\in\mathcal K^0}\!\!I_{i,j}(1\!-\!I_{i,j})\!\!\right)\nonumber\\
    &+\mu\max_{i\in\mathcal K}\left\{\frac{D_iN_{\mathsf{FLOP}}}{f_i}+\frac{B}{\sum_{j\in\mathcal K^0}r_{i,j}I_{i,j}}\right\}\label{eq:obj}\\
    \mathrm{s.t.} \,
    &D_iN_{\mathsf{FLOP}}t_{i,j}^2+Bs_{i,j}^2-\frac{D_jN_{\mathsf{FLOP}}}{f_j}\leq 0,\,i\in\mathcal K,j\in\mathcal K,\label{eq:dc1}\\
    &\frac{1}{f_i}-\frac{t_{i,j}^2}{I_{i,j}}\leq 0,\,i\in\mathcal K, j\in\mathcal K,\label{eq:dc2}\\
    &\frac{1}{\sum_{j\in\mathcal K^0}r_{i,j}I_{i,j}}-\frac{s_{i,j}^2}{I_{i,j}}\leq 0,\,i\in\mathcal K, j\in\mathcal K,\label{eq:dc3}\\
    &(\ref{eq:node_topo_1}), (\ref{eq:node_topo_2}), (\ref{eq:continuousI}), \text{ and } (\ref{eq:freqconst})\nonumber.
\end{align}

Now it can be seen that problem ($\mathbf {P3}$) is a DC problem. Thus we can construct a series of convex surrogate functions for all the concave terms using first-order Taylor expansion. Specifically, we have the following observations:
\begin{enumerate}
    \item The only concave term in the objective function is $-I^2_{i,j}$. Its convex surrogate function at $\mathbf I^{(t)}$ is given by $-2I_{i,j}^{(t)}I_{i,j}+I_{i,j}^{(t)^2}$;
    \item The concave terms in constraints (\ref{eq:dc1}), (\ref{eq:dc2}) and (\ref{eq:dc3}) are $-\frac{D_jN_{\mathsf{FLOP}}}{f_j}$, $-\frac{t_{i,j}^2}{I_{i,j}}$ and $-\frac{s_{i,j}^2}{I_{i,j}}$, respectively. Their convex surrogate functions at $(\mathbf I^{(t)}$, $\mathbf f^{(t)}, \mathbf T^{(t)},\mathbf S^{(t)})$ are $\frac{D_jN_{\mathsf{FLOP}}f_j}{f_j^{(t)^2}}-\frac{2D_jN_{\mathsf{FLOP}}}{f_j^{(t)}}$, $\left[-\frac{2t_{i,j}^{(t)}}{I_{i,j}^{(t)}},\frac{t_{i,j}^{(t)^2}}{I_{i,j}^{(t)^2}}\right][t_{i,j}-t_{i,j}^{(t)},I_{i,j}-I_{i,j}^{(t)}]\trans-\frac{t_{i,j}^{(t)^2}}{I_{i,j}^{(t)}}$, and $\left[-\frac{2s_{i,j}^{(t)}}{I_{i,j}^{(t)}},\frac{s_{i,j}^{(t)^2}}{I_{i,j}^{(t)^2}}\right][s_{i,j}-s_{i,j}^{(t)},I_{i,j}-I_{i,j}^{(t)}]\trans\!-\!\frac{s_{i,j}^{(t)^2}}{I_{i,j}^{(t)}}$, respectively;
    \item All other constraints of problem ($\mathbf {P3}$) are linear.
\end{enumerate}

Replacing the concave terms in the objective function and the DC constraints of problem ($\mathbf {P3}$) with the linearly expanded terms, and applying SCA algorithm, we have a sequence of convex optimization problems:
\begin{align}
    (\textbf{P}^{(t+1)})\nonumber\\
    \min_{\mathbf I,\mathbf f,\mathbf T,\mathbf S} \, &\sum_{i\in \mathcal K}\bigg(\kappa_i D_iN_{\mathsf{FLOP}}f_i^2+\frac{BP}{\sum_{j\in\mathcal K^0}r_{i,j}I_{i,j}}\nonumber\\
    &+\frac{1}{\beta}\sum_{j\in\mathcal K^0}(I_{i,j}-2I_{i,j}^{(t)}I_{i,j}+I_{i,j}^{(t)^2})\bigg)\nonumber\\&+\mu\max_{i\in\mathcal K}\left\{\frac{D_iN_{\mathsf{FLOP}}}{f_i}+\frac{B}{\sum_{j\in\mathcal K^0}r_{i,j}I_{i,j}}\right\}\\
    \mathrm{s.t.} \,
    &D_iN_{\mathsf{FLOP}}t_{i,j}^2\!+\!Bs_{i,j}^2\!+\!\frac{D_jN_{\mathsf{FLOP}}f_j}{f_j^{(t)^2}}\!-\!\frac{2D_jN_{\mathsf{FLOP}}}{f_j^{(t)}}\!\leq\! 0,\nonumber\\
    &\qquad\qquad\qquad\qquad\qquad\qquad i\in\mathcal K,j\in\mathcal K,\label{eq:expand1}\\
    &\frac{1}{f_i}+\left[-\frac{2t_{i,j}^{(t)}}{I_{i,j}^{(t)}},\frac{t_{i,j}^{(t)^2}}{I_{i,j}^{(t)^2}}\right][t_{i,j}-t_{i,j}^{(t)},I_{i,j}-I_{i,j}^{(t)}]\trans\nonumber\\
    &\qquad\qquad\qquad-\frac{t_{i,j}^{(t)^2}}{I_{i,j}^{(t)}}\leq 0,\,i\in\mathcal K, j\in\mathcal K,\label{eq:expand2}\\
    &\frac{1}{\sum_{j\in\mathcal K^0}r_{i,j}I_{i,j}}\!\!+\!\!\left[\!-\!\frac{2s_{i,j}^{(t)}}{I_{i,j}^{(t)}},\!\frac{s_{i,j}^{(t)^2}}{I_{i,j}^{(t)^2}}\!\right]\!\![s_{i,j}\!-\!s_{i,j}^{(t)},\!I_{i,j}\!-\!I_{i,j}^{(t)}]\trans\nonumber\\
    &\qquad\qquad\qquad-\frac{s_{i,j}^{(t)^2}}{I_{i,j}^{(t)}}\!\leq\! 0,\,i\in\mathcal K, j\in\mathcal K,\label{eq:expand3}\\
    &(\ref{eq:node_topo_1}), (\ref{eq:node_topo_2}), (\ref{eq:continuousI}), \text{ and } (\ref{eq:freqconst}),\nonumber
\end{align}
where $(\textbf{P}^{(t+1)})$ is the optimization problem in the $(t+1)$-th iteration of the SCA algorithm, and $\mathbf I^{(t)}$, $\mathbf f^{(t)}$, $\mathbf T^{(t)}$ and $\mathbf S^{(t)}$ are the optimal solution of problem $(\textbf{P}^{(t)})$. Note that each $(\textbf{P}^{(t)})$ is a convex optimization problem and can be solved via off-the-shelf toolbox (e.g. CVX Mosek). According to Theorem 1 of \cite{BMarks1978InnerApprox}, any limit point of the sequence $\{(\mathbf I^{(t)}, \mathbf f^{(t)}, \mathbf T^{(t)}, \mathbf S^{(t)})\}_{t=0,1,\cdots}$ is the KKT solution to the problem ($\mathbf{P3}$) for any feasible starting point $(\mathbf I^{(0)}, \mathbf f^{(0)}, \mathbf T^{(0)}, \mathbf S^{(0)})$.

\subsection{Two-Stage Initialization and Summary of Our Algorithm}
It is worth emphasizing that the SCA algorithm should be carefully initialized. An intuitive way to initialize $(\textbf{P}^{(t)})$ is to make all the edge devices directly connected to the edge server, i.e., to initialize $\mathbf I^{(0)}=[\mathbf 0_{K\times K};\mathbf 1_{K}]$, where $\mathbf 0_{K\times K}$ represents the $K\times K$ matrix with all-zero entries, and $\mathbf 1_{K}$ represents a $K$ column vector with all-one entries.
However, such an intuitive initialization easily leads to slow convergence of the SCA algorithm. Hence, we propose a two-stage initialization method to accelerate the convergence. Specifically, at the first stage, we initialize the solution by connecting all the devices to the edge server, i.e., $\mathbf I^{(0)}=[\mathbf 0_{K\times K};\mathbf 1_{K}]$. Accordingly, $\mathbf f$ is initialized as $\mathbf f^{(0)}=(f_{\min}+f_{\max})/2*\mathbf 1_K$, $\mathbf T$ and $\mathbf S$ are initialized as $t_{i,j}^{(0)^2}=\frac{I_{i,j}^{(0)}}{f_i^{(0)}}$ and $s_{i,j}^{(0)^2}=\frac{I_{i,j}^{(0)}}{\sum_{j\in\mathcal K^0}I_{i,j}^{(0)}r_{i,j}}$ for all $i\in\mathcal K$, $j\in\mathcal K$.
At the second stage, we solve a series of convex optimization problem $(\textbf{P}_{\text{init}}^{(t)})$ which is similar to $(\textbf{P}^{(t)})$, but removes the penalty term in the objective function of $(\textbf{P}^{(t)})$, using the initial point at the first stage,
\begin{align}
    (\textbf{P}_{\text{init}}^{(t+1)})\quad
    \min_{\mathbf I,\mathbf f,\mathbf T,\mathbf S} \quad &\sum_{i\in \mathcal K}\left(\kappa_i D_iN_{\mathsf{FLOP}}f_i^2+\frac{BP}{\sum_{j\in\mathcal K^0}r_{i,j}I_{i,j}}\right)\nonumber\\
    &+\mu\max_{i\in\mathcal K}\left\{\frac{D_iN_{\mathsf{FLOP}}}{f_i}+\frac{B}{\sum_{j\in\mathcal K^0}r_{i,j}I_{i,j}}\right\}\\
    \mathrm{s.t.} \quad
    &(\ref{eq:expand1}), (\ref{eq:expand2}), (\ref{eq:expand3}), (\ref{eq:node_topo_1}), (\ref{eq:node_topo_2}), (\ref{eq:continuousI}), \text{ and } (\ref{eq:freqconst})\nonumber.
\end{align}
We iteratively solve $(\textbf{P}_{\text{init}}^{(t)})$ until convergence, and the output of this problem is used as the initial point for solving $(\textbf{P}^{(t)})$.

As a summary, the penalty-based SCA optimization procedure to solve $(\mathbf {P3})$ is given in Algorithm 1.
\begin{algorithm}[t!]
    \caption{Penalty-based SCA algorithm for optimizing $\mathbf I$ and $\mathbf f$}\label{alg:SCA}
    \begin{algorithmic}[1]
    \State \textbf{\# INITIALIZATION}
    \State \textbf{Stage 1:}
    \State \textbf{Initialize} $\mathbf I^{(0)}=[\mathbf 0_{K\times K};\mathbf 1_{K}]$, $\mathbf f^{(0)}=(f_{\min}+f_{\max})/2*\mathbf 1_K$, $t_{i,j}^{(0)^2}=\frac{I_{i,j}^{(0)}}{f_i^{(0)}}$ and $s_{i,j}^{(0)^2}=\frac{I_{i,j}^{(0)}}{\sum_{j\in\mathcal K^0}I_{i,j}^{(0)}r_{i,j}}$ for all $i\in\mathcal K$, $j\in\mathcal K$.
    \State \textbf{Stage 2:}
    \State \textbf{Repeat}
    \State \qquad Update $\mathbf I^{(t+1)}$, $\mathbf f^{(t+1)}$, $\mathbf T^{(t+1)}$, $\mathbf S^{(t+1)}$ by solving \indent $(\textbf{P}_{\text{init}}^{(t+1)})$ via Mosek.
    \State \qquad$t\leftarrow t+1.$
    \State \textbf{Until} convergence.
    \State \textbf{Initialize} $\mathbf I_{\text{init}}=\mathbf I^{(t)}$, $\mathbf f_{\text{init}}=\mathbf f^{(t)}$, $\mathbf T_{\text{init}}=\mathbf T^{(t)}$, $\mathbf S_{\text{init}}=\mathbf S^{(t)}$.
    \State \textbf{\# OPTIMIZATION}
    \State \textbf{Repeat}
    \State\qquad Update $\mathbf I^{(t+1)}$, $\mathbf f^{(t+1)}$, $\mathbf T^{(t+1)}$, $\mathbf S^{(t+1)}$ by solving \indent $(\textbf{P}^{(t+1)})$ via Mosek.
    \State\qquad $t\leftarrow t+1.$
    \State \textbf{Until} convergence.
    \State \textbf{Output} $\mathbf I^*=\text{round}(\mathbf I^{(t)})$, $\mathbf f^*=\mathbf f^{(t)}$.
    \end{algorithmic}
\end{algorithm}
\begin{Remark}[Complexity analysis]\label{rem:complexity}
    For the SCA algorithm, each $(\textbf{P}^{(t+1)})$ involves $K(K+1)+K+2K^2=3K^2+2K$ primal variables and $3K^2+2K+1+K(K+1)+K=4K^2+4K+1$ dual variables. Therefore, the worst-case complexity for solving $(\textbf{P}^{(t+1)})$ is $\mathcal O\left((7K^2+6K+1)^{3.5}\right)$ \cite{BenTal2001cvx}. In turn, the total complexity for solving problem ($\mathbf{P3}$) is $\mathcal O\left(N_{\text{iter}}(7K^2+6K+1)^{3.5}\right)$, where $N_{\text{iter}}$ is the number of SCA iterations and its value is around $3\sim 5$ as shown in the simulation.
\end{Remark}

\section{Imitation-learning based Method}
In the last section, we have developed a penalty-based method to transform the original optimization problem into a continuous one, and exploit SCA to solve the continuous problem. As mentioned in Remark \ref{rem:complexity}, however, the complexity of the SCA algorithm is still high especially when the number of devices is large, which makes it impractical for real-time implementation. To address this challenge, we will propose an imitation-learning based method for faster topology and computing speed design.

Imitation learning is a machine learning paradigm that trains smart agents by learning from demonstrations \cite{Yu2020IE}. The core advantage of imitation learning comes from its offline training and online decision making manner. Thus, the trained DNN model can be efficiently applied to make real-time decisions. In this work, we adopt the optimal topologies and computing speeds with respect to various wireless channels as the demonstrations, which is generated by the penalty-based algorithm and collected as training data samples. Then, we use these high-quality demonstrations to train our imitation learning DNN model in an offline manner. Afterwards, the trained DNN model can be deployed either on the edge server or the edge devices to imitate the optimal decision pattern (topology and computing speed) and perform real-time inference.

\subsection{Imitation Learning DNN Design}
In this subsection, we elaborate the design of the DNN for imitation learning. In our considered system, the wireless channels are the parameters that change relatively fast. It might be prohibitive to run the penalty-based algorithm and make decisions every frame. The imitation learning DNN takes the wireless channel matrix $\mathbf H=\{H_{i,j}|i\in\mathcal K,j\in\mathcal K^0\}$ as the input, and its output consists of the topology matrix $\mathbf I$ and computing speed $\mathbf f$. However, directly designing a DNN that outputs $\mathbf I$ and $\mathbf f$ is impractical, since the the dimension of output $(K(K+1)+K)$ is comparable to that of the input features $(K(K+1))$. To this end, we design two DNNs for each edge device $i\in\mathcal K$: $\text{NN}_i^{\text I}$ which takes $\mathbf H$ as input, $j_i\in\mathcal K^0$ that represents the target receiver of node $i$ as output, and $\text{NN}_i^{\text f}$ which takes $\mathbf H$ as input, $f_i$ as output. An illustration of $\text{NN}_i^{\text I}$ and $\text{NN}_i^{\text f}$ is shown in Fig. \ref{fig:nns}, where Fig. \ref{fig:nni} shows the classification network $\text{NN}_i^{\text I}$ for computing node $i$'s target receiver node index $j_i$, and Fig. \ref{fig:nnf} shows a regression network $\text{NN}_i^{\text f}$ for obtaining the computing speed $f_i$ of node $i$.
\begin{figure*}[t!]
    \centering
    \includegraphics[width=\linewidth]{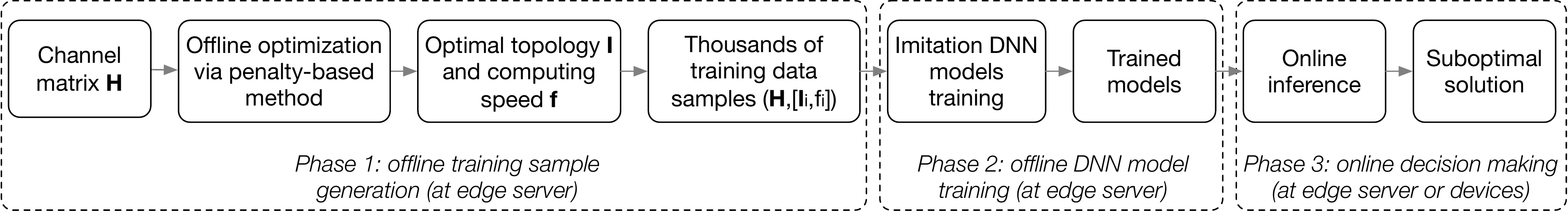}
    \caption{Workflow of the imitation-learning based framework.}
    \label{fig:IL_flowchart}
\end{figure*}
\begin{figure}
    \centering
    \begin{subfigure}[b]{0.45\textwidth}
      \centering
      \includegraphics[width=\textwidth]{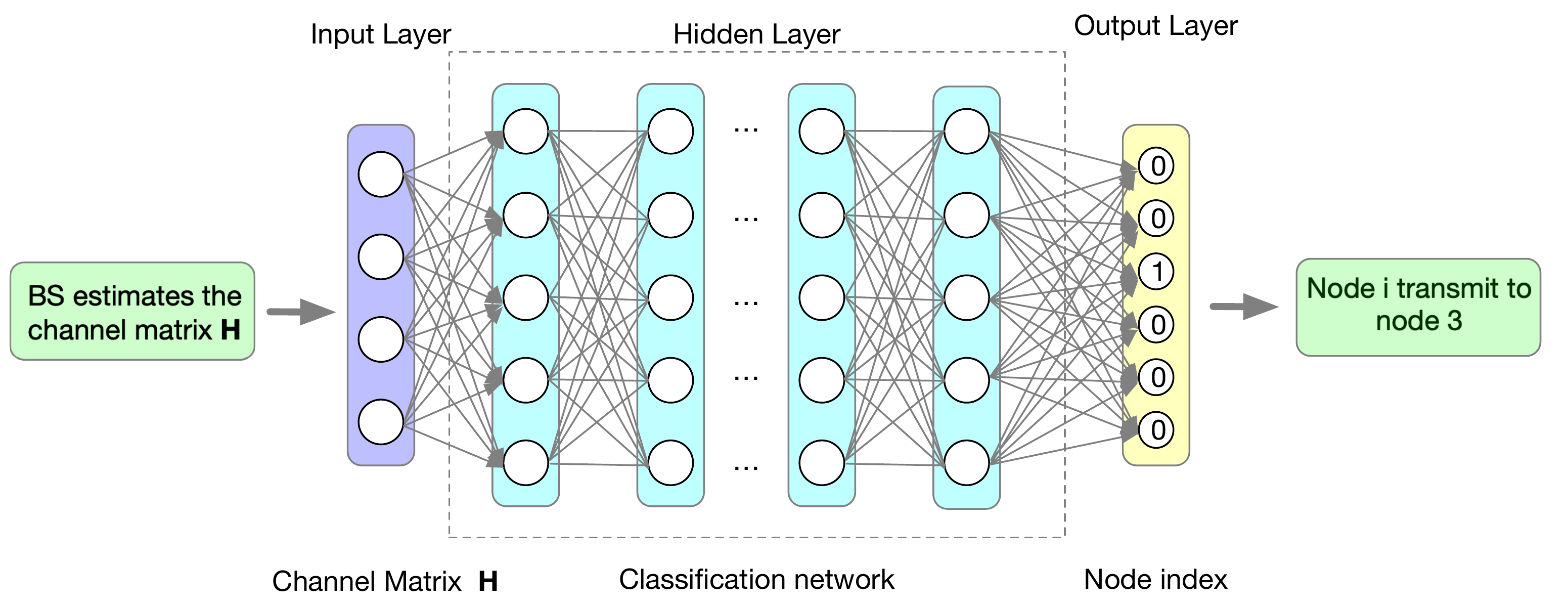}
      \caption{$\text{NN}_i^{\text{I}}$ outputting node index. }
      \label{fig:nni}
    \end{subfigure}
    \hfill
    \begin{subfigure}[b]{0.45\textwidth}
      \centering
      \includegraphics[width=\textwidth]{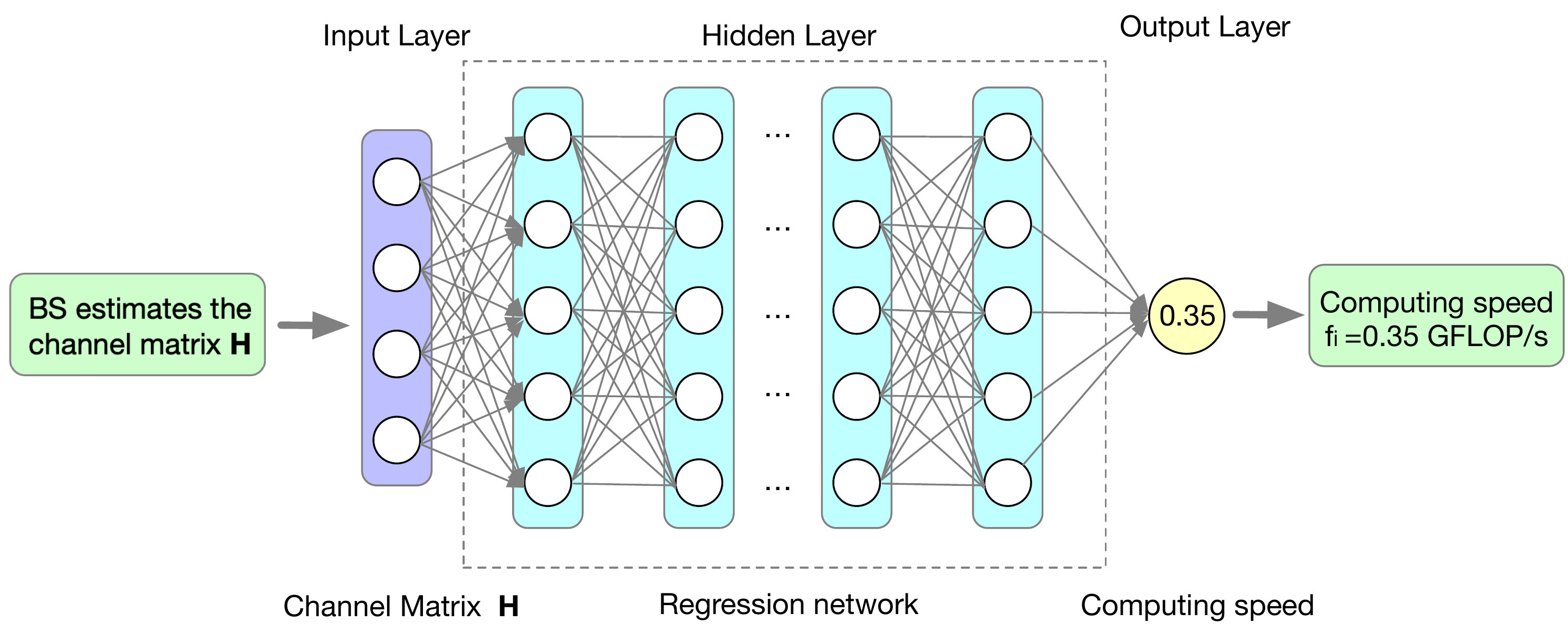}
      \caption{$\text{NN}_i^{\text{f}}$ outputting computing speed.}
      \label{fig:nnf}
    \end{subfigure}
    \caption{Illustration of the DNNs for imitation learning.}
    \label{fig:nns}
\end{figure}

\subsection{Workflow of Imitation Learning}\label{sec:workflow}
Fig. \ref{fig:IL_flowchart} gives a detailed illustration of the workflow of our imitation-learning based method.
The imitation-learning based method consists of three phases: 1) \textit{Offline training sample generation}, 2) \textit{Offline imitation DNN training}, 3) \textit{Online decision making}, which are elaborated as follows.

1) \textit{Offline training sample generation:}
In this phase, we use the penalty-based method to generate training data samples $(\mathbf H, [\mathbf I_i,f_i])$, where $\mathbf I_i$ is the $i$-th row of $\mathbf I$ and is a one-hot vector indicating the index $j_i$ of node $i$'s target receiver. These training data samples are used as the demonstrations to train the imitation learning network.

2) \textit{Offline imitation DNN training:}
In this phase, we train the DNNs by feeding into the training samples generated in phase 1. For the classification network $\text{NN}_i^{\text{I}}$, we use a 4-layer fully connected DNN, where rectified linear units (ReLUs) are used as the activation functions of the hidden layers, and Softmax is used as the activation function of the output layer. For the regression network $\text{NN}_i^{\text{f}}$, we also use a 4-layer fully connected DNN, where Sigmoid activations are used for all the layers. Besides, cross-entropy loss \cite{Goodfellow2016DL} is used as the performance metric for the classification network $\text{NN}_i^{\text{I}}$, and MSE loss \cite{Goodfellow2016DL} is used as the performance metric for the regression network $\text{NN}_i^{\text{f}}$. Moreover, Adam optimizer \cite{Goodfellow2016DL} is used to optimize all the DNNs.

3) \textit{Online decision making:}
After the imitation DNN models are well-trained, they can be deployed to make real-time decisions. Thus, when the wireless channels change, we can feed the channel matrix $\mathbf H$ to the DNNs, which will directly output the target receiving node index and the computing speed of each device $i$ in an online manner.

\subsection{Distributed Implementation}
In our proposed imitation learning framework, we design a classification network $\text{NN}_i^{\text{I}}$ and a regression network $\text{NN}_i^{\text{f}}$ for each device $i$. The trained $\text{NN}_i^{\text{I}}$ and $\text{NN}_i^{\text{f}}$ can be deployed at the local devices to perform distributed inference.
The integrated framework for the federated learning system with optimized topology and computing speed is summarized in Algorithm \ref{alg:distributed}.

\begin{algorithm}[t]
    \caption{Imitation-learning based distributed implementation of TOFEL }\label{alg:distributed}
    \begin{algorithmic}[1]
    \State Edge server initializes the global model $\mathbf w^{(0)}$, and broadcasts $\mathbf w^{(0)}$ and $\mathbf H^{(0)}$ to the edge devices.
    \State \textbf{for} $t=0,1,\cdots$, \textbf{do}\\
    \qquad \textbf{for} node $i\in\mathcal K$ in parallel \textbf{do}
    \State \qquad\qquad Node $i$ chooses a node $j_i$ for transmission, and \indent\indent decide its computing speed $f_i$ by feeding $\mathbf H^{(t)}$ \indent\indent to $\text{NN}_i^{\text{I}}$ and $\text{NN}_i^{\text{f}}$.\\
    \qquad\qquad Node $i$ computes  $\mathbf g_i^{(t)}=\nabla F_i(\mathbf w^{(t)})$ with \indent\indent CPU  speed $f_i$, and aggregates the gradi-\indent\indent ents  from the nodes $k$ with $I_{k,i}=1$, i.e., \indent\indent $\tilde{\mathbf g}_i=\frac{\sum_{k\in\mathcal K}I_{k,i}\tilde D_k\tilde{\mathbf g}_k+D_i\mathbf g_i}{\sum_{k\in\mathcal K}I_{k,i}\tilde D_k+D_i}$.\\
    \qquad\qquad Node $i$ transmits its aggregated gradient $\tilde{\mathbf g}_i$ to \indent\indent node $j_i$. \\
    \qquad\textbf{end for}\\
    \qquad Edge server aggregates the gradient vectors from \indent node $i\in\{i|I_{i,K+1}=1\}$ to obtain $\mathbf g^{(t)}$.\\
    \qquad Edge server performs gradient descent to update the \indent global model via $\mathbf w^{(t+1)}=\mathbf w^{(t)}-\eta\mathbf g^{(t)}$. \\
    \qquad Edge server estimates channel $\mathbf H^{(t+1)}$, and broadcasts \indent $\mathbf w^{(t+1)}$ and $\mathbf H^{(t+1)}$ to the devices.\\
    \textbf{end for}
    \State \textbf{Output} Global model $\mathbf w=\mathbf w^{(t)}$.
    \end{algorithmic}
\end{algorithm}

\begin{Remark}
Training the imitation learning network may be computation-intensive, especially generating the large amount of training data samples using the penalty-based method. Fortunately, the time-demanding sample generating and model training processes can be done in high performance servers in an offline manner. Then, the well-trained models can be efficiently deployed at edge devices for online decision making. Also note that imitation learning lies in the domain of supervised learning. The training and inference work in the same scenario. If we want to apply the trained models to a new scenario, say a new federated learning scenario with different numbers of edge devices, we may either train a new  model again, or resort to transfer learning, which is celebrated and powerful since  it is capable of applying the knowledge learned in previous scenarios to new related scenarios.
\end{Remark}

\begin{Remark}
	 As described in \ref{sec:workflow}-II, the topology and frequency are decided by different types of networks, i.e., a classification network and a regression network. Therefore, two separated DNNs are used in imitation learning. It is worth noting that the topology and frequency decisions output from the two DNNs can match well with each other once if the training dataset is sufficiently large so that it mimics the distribution of the ground-truth dataset well.
\end{Remark}

\section{Simulation and Experimental Results}
In this section, we evaluate the performance of our proposed TOFEL scheme, and compare it with benchmark schemes. Also, the effectiveness of our proposed penalty-based method and the imitation-learning based method are verified by simulations.
\begin{figure}[t]
    \centering
    \includegraphics[width=\linewidth]{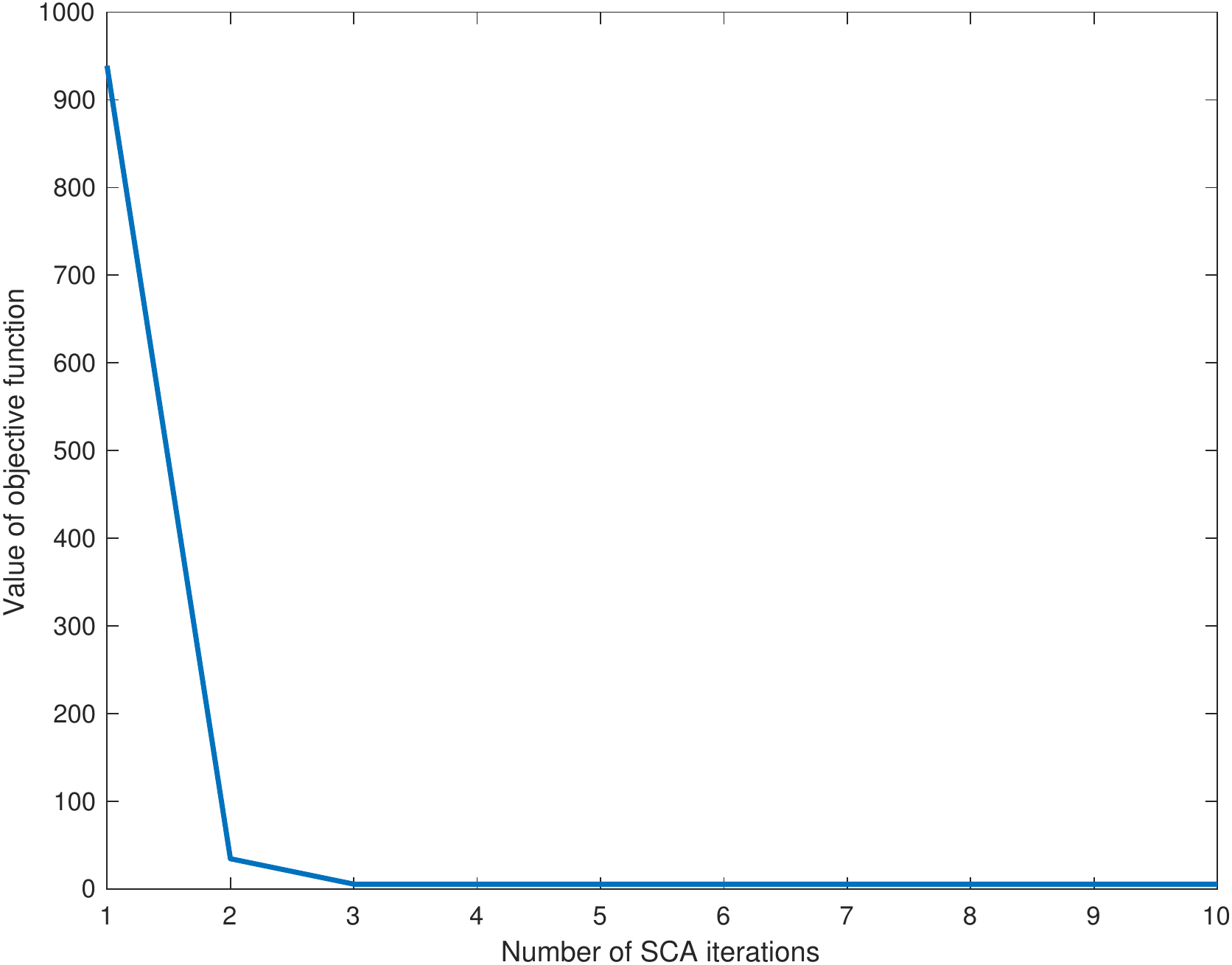}
    \caption{The convergence of the SCA algorithm.}
    \label{fig:SCAconverge}
\end{figure}
\subsection{Simulation Setup}\label{sec:setup}
The simulation settings are as follows unless otherwise specified. A set of $K$ edge devices are uniformly distributed within a cell with radius $200\,\mathrm{m}$. The devices' CPU coefficients $\kappa_i$ are set to $1\times 10^{-28}$. Each edge device is allocated a spectrum band $w_i=180\,\mathrm{KHz}$. Also, the pathloss constant $g_0=-30\,\mathrm{dB}$, the pathloss coefficient $\alpha =3.2$, and the reference distance $d_0=1\,\mathrm{m}$. The noise power is $\sigma^2=10^{-9}\,\mathrm{W}$. The transmit power of the edge devices is $P=100\,\mathrm{mW}$. For the federated learning task, we consider the handwritten digit recognition in our experiment. The celebrated MNIST dataset \cite{MNIST} is used for training a convolutional neural network (CNN). For the MNIST dataset, it consists of $70000$ grayscale images (a training set of $60000$ examples and a test set of $10000$ examples) of handwritten digits, each with $28\times 28$ pixels. Thus, each image needs $28\times 28\times 8+4 = 6276\, \mathrm{bits}$, and $N_{\mathrm{FLOP}}=6276\times 5=31380\,\mathrm{FLOP/sample}$, where we have assumed $5\,\mathrm{FLOP}$s are required to process one-bit of input data. For the CNN, it consists of a $5 \times 5$ convolution layer (with ReLu activation, $32$ channels), a $2 \times 2$ max pooling layer, another $5 \times 5$ convolution layer (with ReLu activation, $64$ channels), a $2 \times 2$ max pooling layer, a fully connected layer with $128$ units (with ReLu activation), and a final softmax output layer (with $10$ outputs). The communication load for transmitting a gradient vector is approximately set as $B=10\,\mathrm{kb}$. Moreover, the computing speed of each edge device may be adjusted between $0.1 \,\mathrm{GFLOP}$s and $1 \,\mathrm{GFLOP}$s. The weighting factor between energy consumption and latency is $\mu=0.5$.

\subsection{Evaluation of Penalty-Based Method}

\textbf{Convergence}. We evaluate the penalty-based SCA method with the simulation setup described above. We first show the convergence of the SCA algorithm via simulation. In this simulation, we set the number of edge devices as $K=5$ unless otherwise specified. The penalty parameter is set as $\beta = 0.0001$. Fig. \ref{fig:SCAconverge} shows the value of the objective function (\ref{eq:obj}) is monotonically decreasing in consecutive SCA iterations, and converges after around $4$ iterations.

\begin{figure}[t]
    \centering
    \begin{subfigure}[b]{0.45\textwidth}
      \centering
      \includegraphics[width=\textwidth]{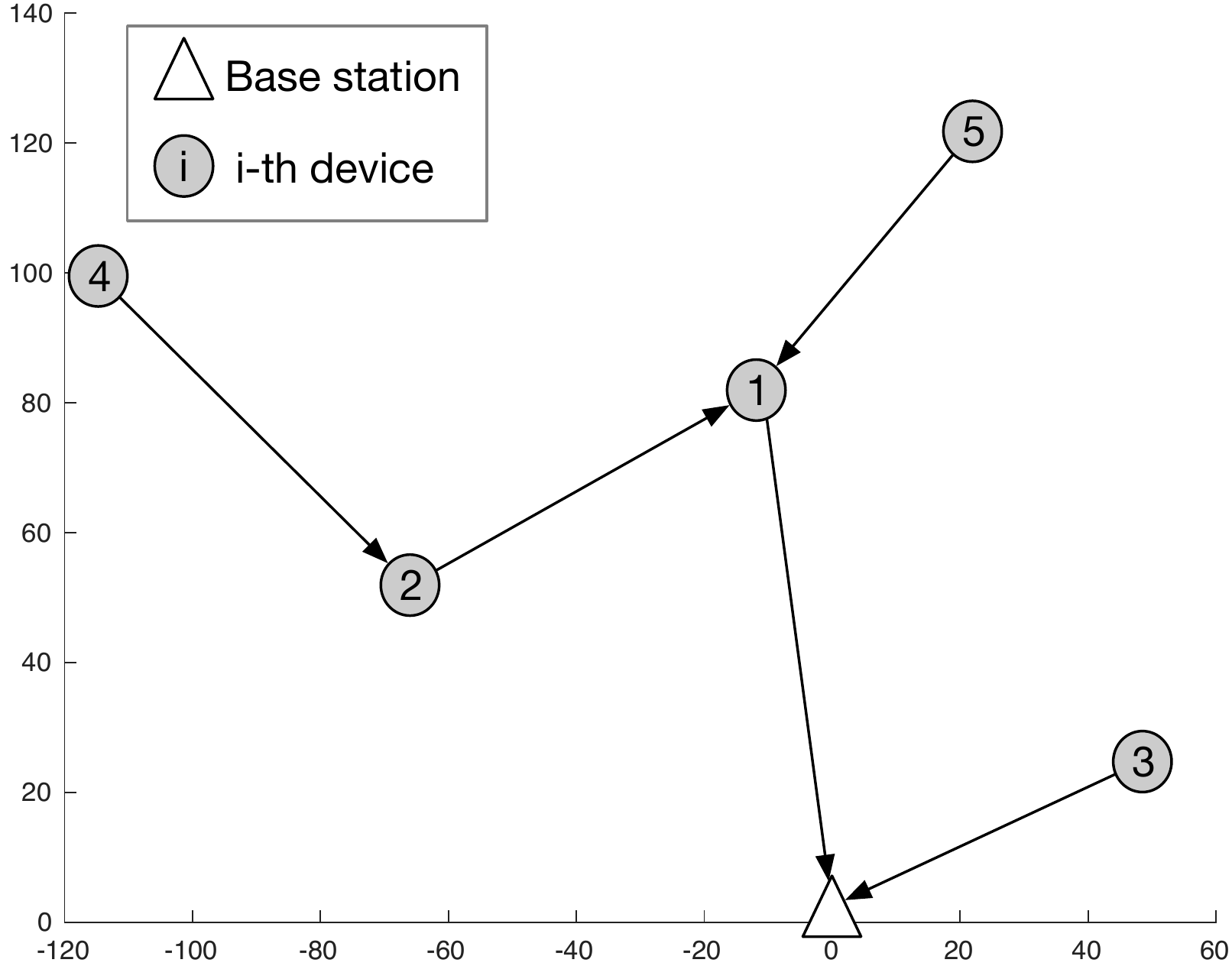}
      \caption{Optimized topology for $5$ devices. }
      \label{fig:topo5}
    \end{subfigure}
    \hfill
    \begin{subfigure}[b]{0.45\textwidth}
      \centering
      \includegraphics[width=\textwidth]{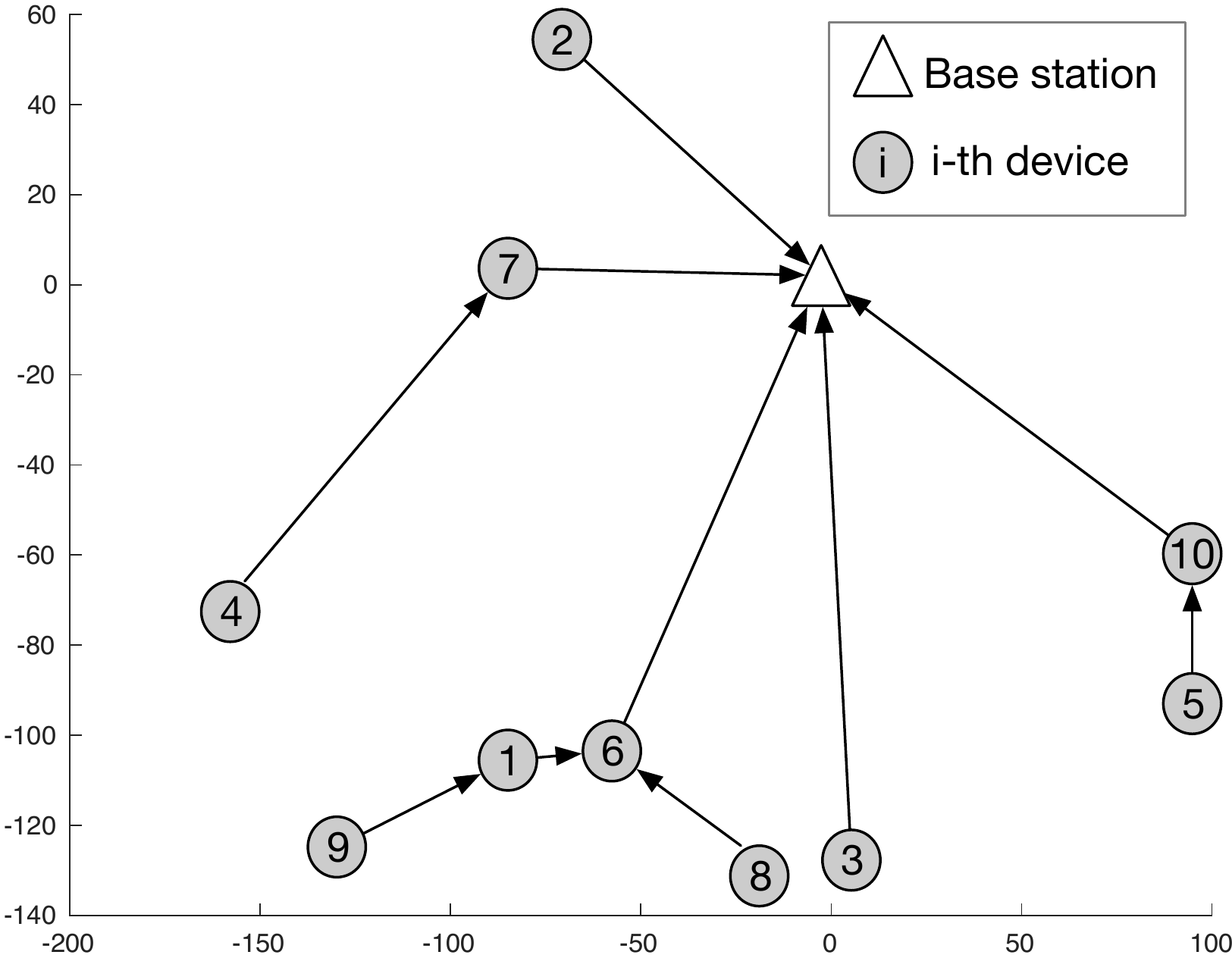}
      \caption{Optimized topology for $10$ devices.}
      \label{fig:topo10}
    \end{subfigure}
    \caption{A realization of the optimized topology of TOFEL systems.}
    \label{fig:topo}
\end{figure}

\textbf{Optimized topology}. In our TOFEL scheme, the topology of the federated learning network is optimized by problem ($\mathbf{P1}$), and a realization of the optimized topology of TOFEL is shown in Fig. \ref{fig:topo}, where Fig. \ref{fig:topo5} and Fig. \ref{fig:topo10} show the optimized topologies of TOFEL systems with $5$ and $10$ edge devices, respectively. It can be seen that unlike conventional FEEL or the hierarchical FEEL with fixed hierarchy as in \cite{Liu2020HFL,Luo2020HFL}, the topology of our proposed TOFEL framework is quite flexible. Any node in the system can be regarded as a candidate node to perform gradient aggregation and forward, and the connection between an edge device and the edge server may span multiple hops. For example, as shown in Fig. \ref{fig:topo5}, the gradient vector of node 4 is first transmitted to and aggregated at node 2, followed by a further aggregation at node 1, and the aggregated gradient vector at node 1 is then transmitted to the edge server for final aggregation. Thus, the connection between node 4 and the edge server contains $3$ hops. This kind of flexibility  allows the federated learning system to better exploit the heterogeneity the wireless channels and the devices' resources. Moreover, since congestion occurs if too many nodes transmit model updates simultaneously to an AF node, we define the degree of a topology as the maximum number of forward links connected to one AF node (similar to the degree of a graph, which is defined as the maximum number of edges connecting to a single vertex). For example, the degrees of the optimized topologies of TOFEL in Fig. \ref{fig:topo5} and Fig. \ref{fig:topo10}  are $2$ and $5$, while they increase to $5$ and $10$ if conventional FEEL is used. Therefore, the proposed topology optimization also alleviates the congestion issue effectively, which is also reflected as the reduced communication latency shown in Fig. \ref{fig:power}.

\begin{figure}[t]
	\centering
	\includegraphics[width=\linewidth]{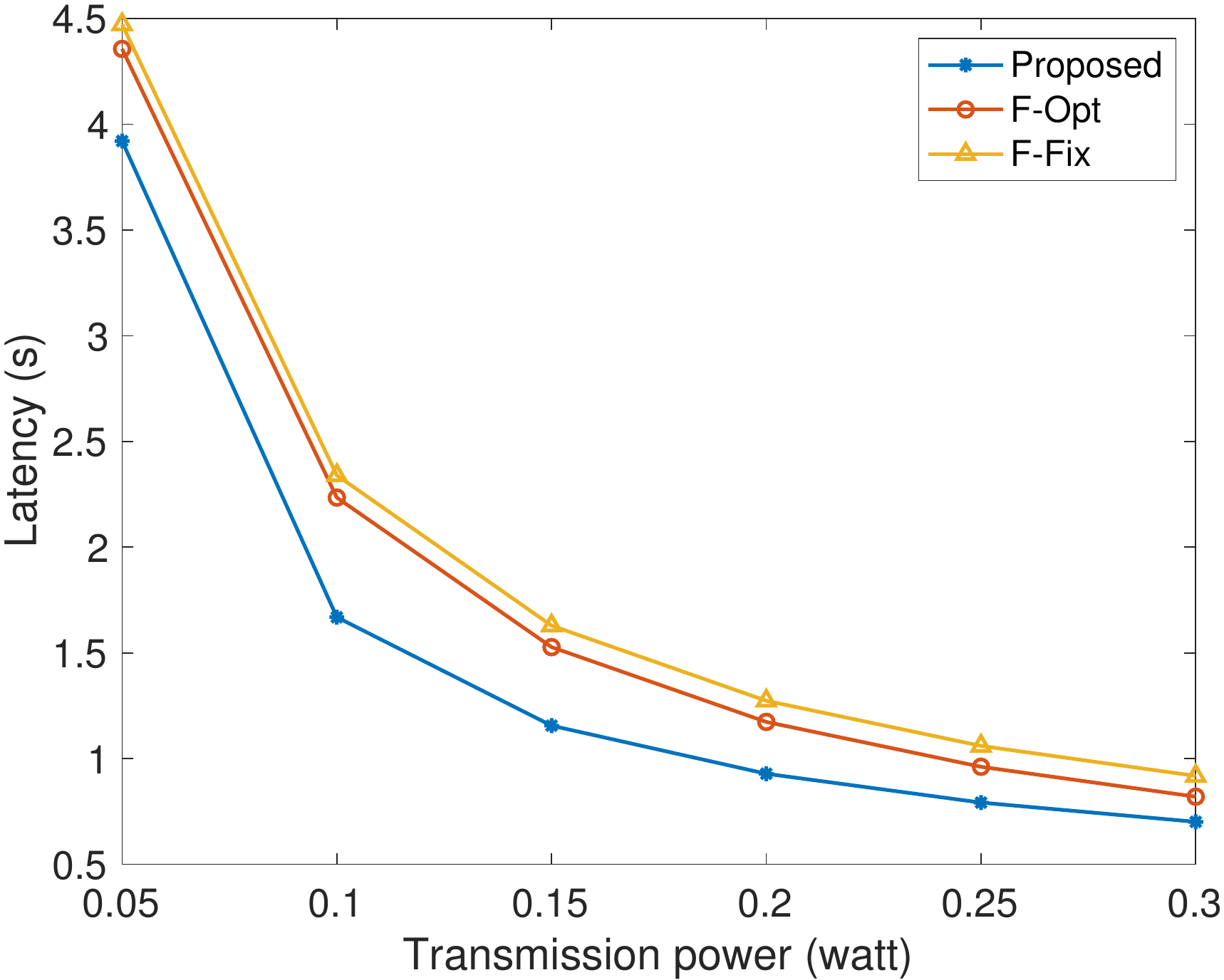}
	\caption{Latency for one-round federated learning iteration.}
	\label{fig:latency}
\end{figure}
\begin{figure}[t]
	\centering
	\includegraphics[width=\linewidth]{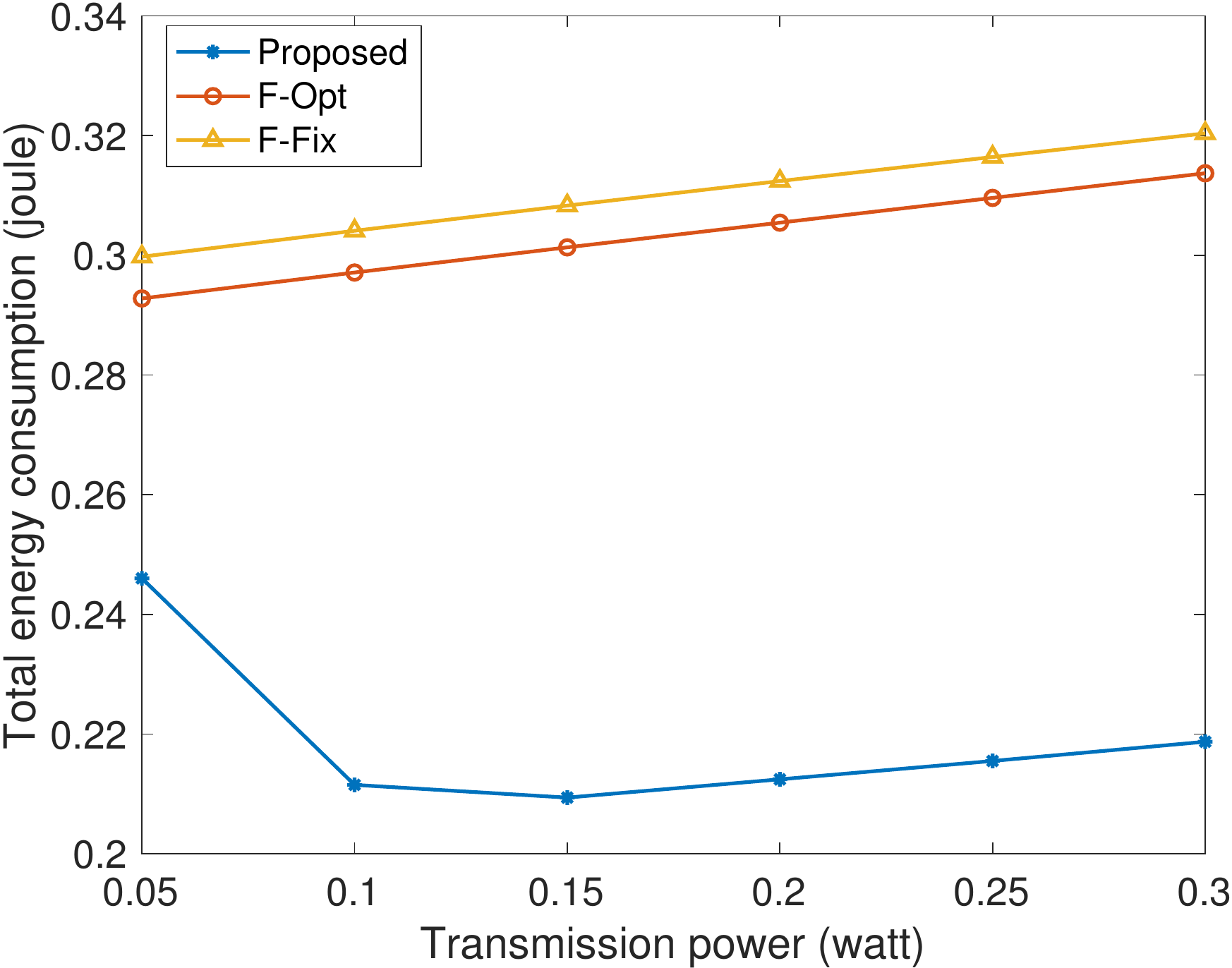}
	\caption{Total energy consumption for one-round federated learning iteration.}
	\label{fig:power}
\end{figure}

\textbf{Comparison with benchmarks}. To demonstrate the performance gain in terms of energy reserving and latency reduction of our proposed TOFEL scheme, we compare it  with two benchmarks.

1) \textit{Flat topology with fixed computing speed} (F-Fix).
In this scheme, all the edge devices directly transmit their local gradients to the edge server for aggregation, and the computing speeds of all the edge devices are set at medium level, i.e., $f_i=(f_{\min}+f_{\max})/2$. Thus, the latency and energy consumption for one-round iteration of federated learning can be directly calculated from equations (\ref{eq:latency}) and (\ref{eq:energy}).

2) \textit{Flat topology with optimized computing speed} (F-Opt). Similarly, in this scheme, all the edge devices directly transmit their local gradients to the edge server for aggregation, but the computing speed of the edge devices are optimized by the following problem:
\begin{align*}
    \mathbf{P^{\text{flat}}}\quad
    \min_{\mathbf f} \quad &\sum_{i\in \mathcal K}\left(\kappa D_iN_{\mathsf{FLOP}}f_i^2+\frac{BP}{r_{i,K+1}}\right)\nonumber\\
    &+\mu\max_{i\in\mathcal K}\left\{\frac{D_iN_{\mathsf{FLOP}}}{f_i}+\frac{B}{r_{i,K+1}}\right\}\\
    \mathrm{s.t.} \quad
    &f_{\min}\leq f_i\leq f_{\max},\, i\in\mathcal K.
\end{align*}
It is easy to verify that this problem is convex, and thus can be efficiently solved via off-the-shelf CVX toolbox such as Mosek.

Fig. \ref{fig:latency} shows the latency of one-round federated learning iteration with different transmit power of the edge devices for our proposed TOFEL scheme and the above mentioned two benchmark schemes. It can be observed that with the increase of the transmit power, the latency of all the schemes decreases. This is intuitive since increasing the transmit power will increase the communication rate. Also can be seen is that optimizing the computing speed of the edge devices can slightly reduce the latency compared with that of fixed computing speed. However, with optimized topology, the latency reduction is much more remarkable, which demonstrates the feasibility and necessity of topology optimization for the federated learning system.

Fig. \ref{fig:power} shows the total energy consumption of one-round federated learning iteration with different transmit power for different schemes. It can be seen that the energy consumption of our proposed scheme is significantly lower than that of the benchmark schemes. Moreover, with the increase of transmit power, the total energy consumptions monotonically increase for the two benchmark schemes. However, the total energy consumption of our proposed scheme will first decrease and then slowly increase. The reasons are two-fold. On one hand, increasing the transmit power will shorten the transmission time. Thus, increasing the transmit power will not necessarily lead to the increase of communication energy consumption. On the other hand, the total energy consumptions is more dominated by communication for the two benchmarks due to the straggler effect, while in our proposed scheme, there exists a trade-off between reducing communication and computation energy consumptions. In this case, increasing the transmit power leads to less communication time and more computation time, which in turn reduces the computation energy.  Hence, if the reduction of computation energy consumption is larger than the increase of communication energy consumption, the total energy consumption would be saved.

To see the impact of $\beta$ on the penalty-based method, we simulate the case of $K=5$ and the cost function value versus the value of $\beta$ is shown in Fig.~9. The effect of different choices of $\beta \in (0,1]$ is concluded as below.
    \begin{itemize}
    	\item When $\beta$ increases from $0$ to $1$, the problem (${\bf P2}$) becomes more distinct from (${\bf P1}$).
    	\item When $\beta$ reduces from $1$ to $0$, (${\bf P2}$) becomes more similar to (${\bf P1}$).
    	\item For the extreme case where $\beta \to 0$, (${\bf P1}$) and (${\bf P1}$) are equivalent.
    \end{itemize}

\begin{figure}[!t]
\centering
\includegraphics[width=0.45\textwidth]{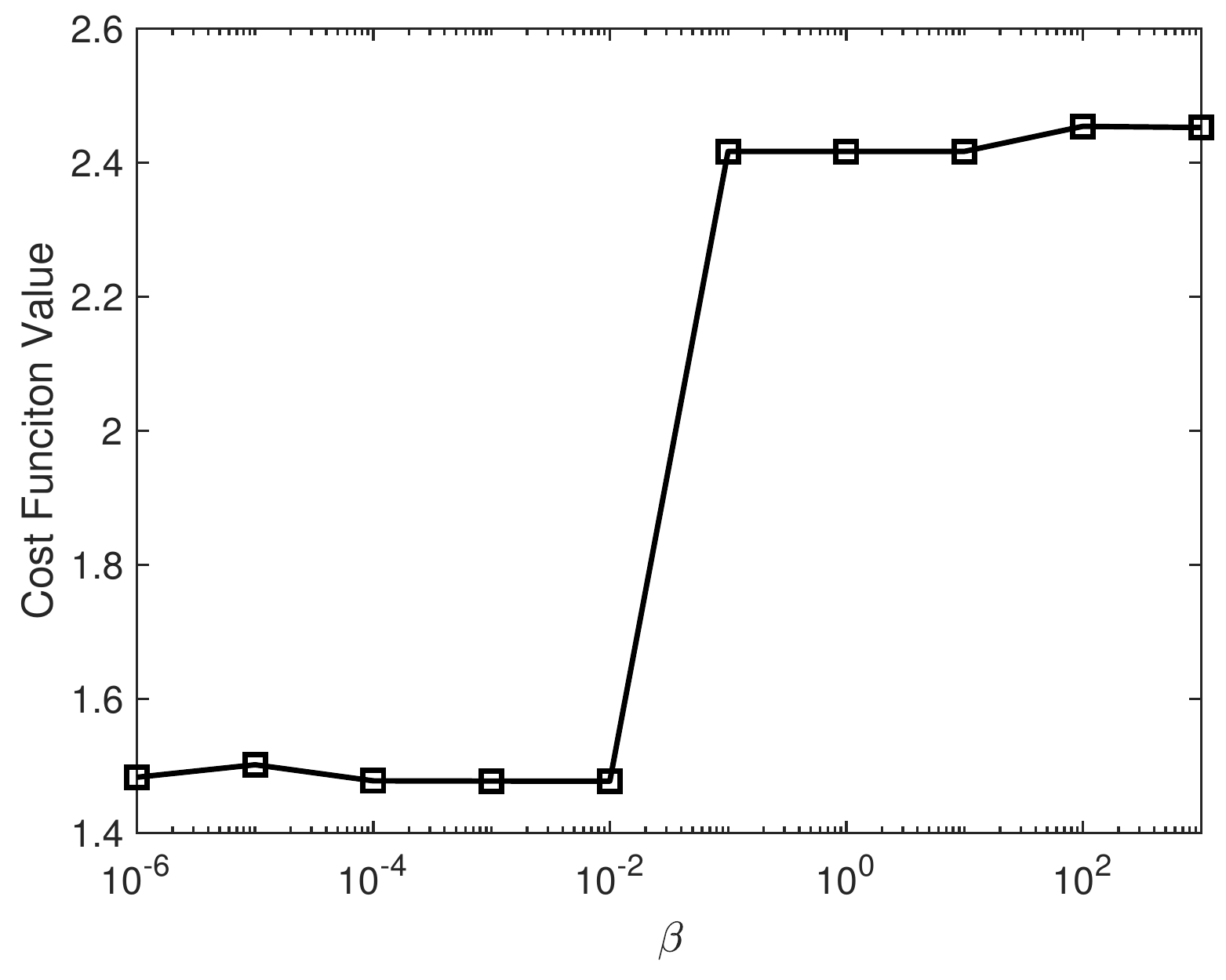}
\caption{Cost function value versus $\beta$ when $K=5$.}
\label{fig_sim}
\vspace{20pt}
\end{figure}

To further save the communication latency, we consider the channel aware scheduling policy which avoids selecting stragglers with weak channels.
Fig.~10 compares the TOEFL schemes with and without device selection when $K=5$.
It can be seen that the proposed TOFEL framework with device selection automatically removes two stragglers far from the base station, thereby significantly reducing the cost, delay, and power.
But this also leads to a smaller number of participating devices in the federated learning group, which may in turn degrade the learning performance.
In the perfect case when the datasets at stragglers do not provide extra information compared with datasets at other users, the device selection scheme can save up to $40\%$ computation power and communication delay as shown in Fig.~10 while guaranteeing the same federated learning performance.

\begin{figure}[!t]
\centering
\includegraphics[width=0.45\textwidth]{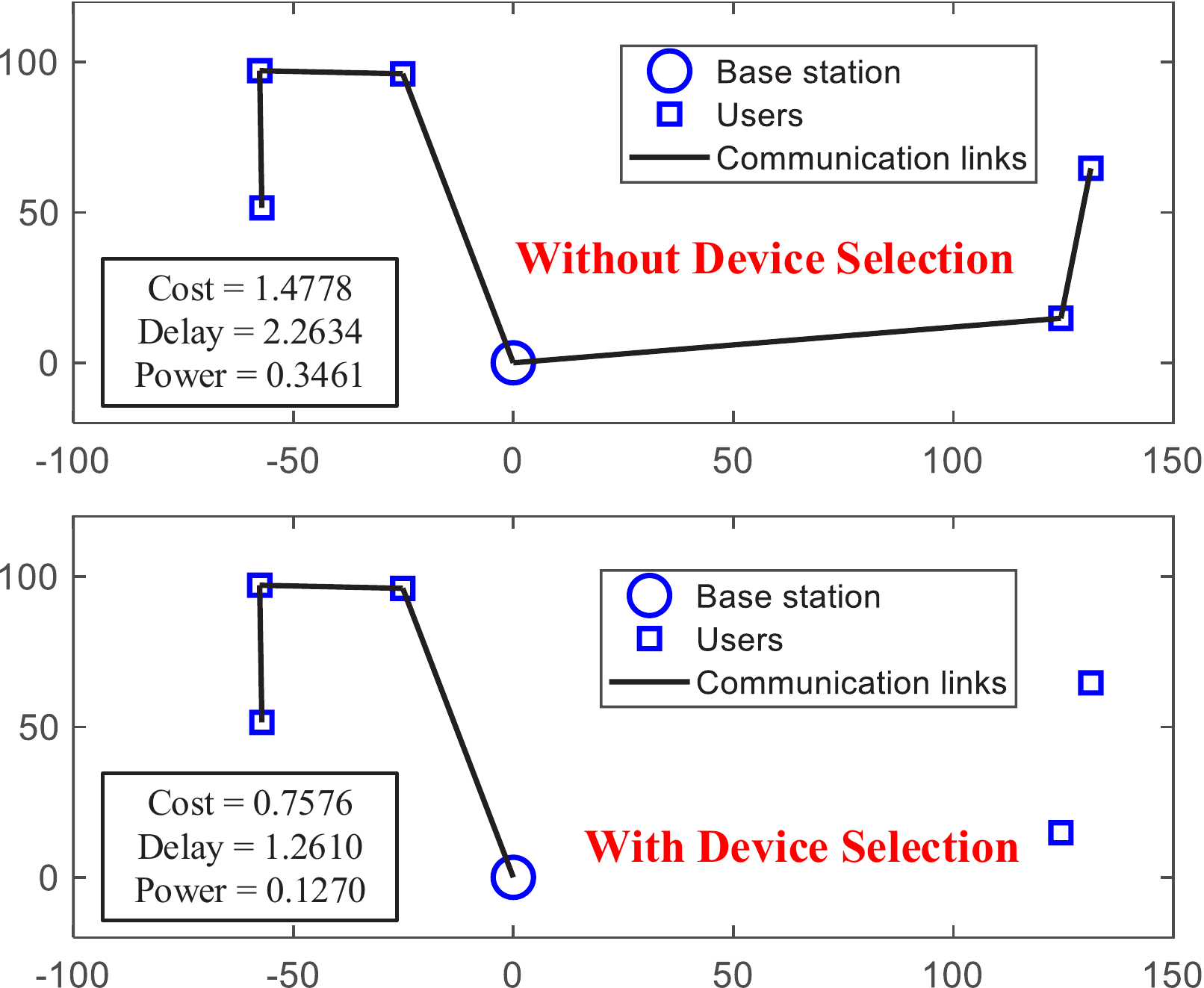}
\caption{Comparison between TOEFL schemes with and without device selection when $K=5$.}
\label{fig_sim2}
\vspace{20pt}
\end{figure}

\subsection{Evaluation of imitation-learning based Method}
\textbf{Training data generation and cleaning}. First of all, we generate $4000$ training samples using the penalty-based SCA method with the simulation setup described in Section \ref{sec:setup}, out of which $3000$ samples are used for training and $1000$ are used for testing. Before feeding the training data into the imitation networks, the data are preprocessed to be more suitable for DNN training. Specifically, we first remove the diagonal entries of the channel matrix $\mathbf H$ since a node cannot transmit to itself and vectorize $\mathbf H$, followed by a normalization process, such that the entries of the processed $\mathbf H$ are all between $0\sim 1$.

\textbf{Imitation DNN design}. As mentioned in Section \ref{sec:workflow}, we use fully-connected networks as the imitation learning networks. Specifically, each $\text{NN}_i^{\text{I}}$ consists of an input layer with $25$ units, two hidden layers, each with $256$ units and each unit with a ReLU activation, and an output layer with Softmax activation. $\text{NN}_i^{\text{I}}$'s are trained with cross-entropy loss and Adam optimizer with learning rate $0.001$. The batch size is set as $64$. Moreover, each $\text{NN}_i^{\text{f}}$ consists of an input layer with $25$ units, a hidden layer with $32$ units and Sigmoid activations, a hidden layer with $16$ units and Sigmoid activations, followed by an output layer. $\text{NN}_i^{\text{f}}$'s are trained with MSE loss and Adam optimizer with learning rate $0.001$. Also, the batch size is set as $64$.

\textbf{Performance of imitation learning}. All the imitation learning DNNs are trained $100$ epochs. Fig. \ref{fig:imitationlearning} shows the learning performance of the imitation learning DNNs of device 1. The learning results of other devices are similar and thus omitted here for conciseness. Fig. \ref{fig:classify} shows the training and test accuracy for $\text{NN}_1^{\text{I}}$, which is a classification network outputting the target receiving node index $j_1$ of node 1. It can be seen that $\text{NN}_1^{\text{I}}$ achieve around $95\%$ training accuracy and $90\%$ test accuracy. Fig. \ref{fig:regress} shows the training and test losses for $\text{NN}_1^{\text{f}}$, which is a regression network outputting the computing speed of device 1. It can be observed that $\text{NN}_1^{\text{f}}$ achieves a training loss of $0.01$ and test loss of $0.05$. Moreover, we record the running time of the penalty-based method and imitation-learning based method. The running time of penalty-based method is around $148\,\mathrm{s}$, while the inference time of the imitation-learning based method is just about $7.5\,\mathrm{ms}$. We can see that the imitation-learning based method can reduce the decision time by more than $10^{4}$ times. Notice that the inference time can be further reduced with GPU computing. Thus, imitation learning method can meet the requirement of real-time decision making.
\begin{figure}
    \centering
    \begin{subfigure}[b]{0.45\textwidth}
      \centering
      \includegraphics[width=\textwidth]{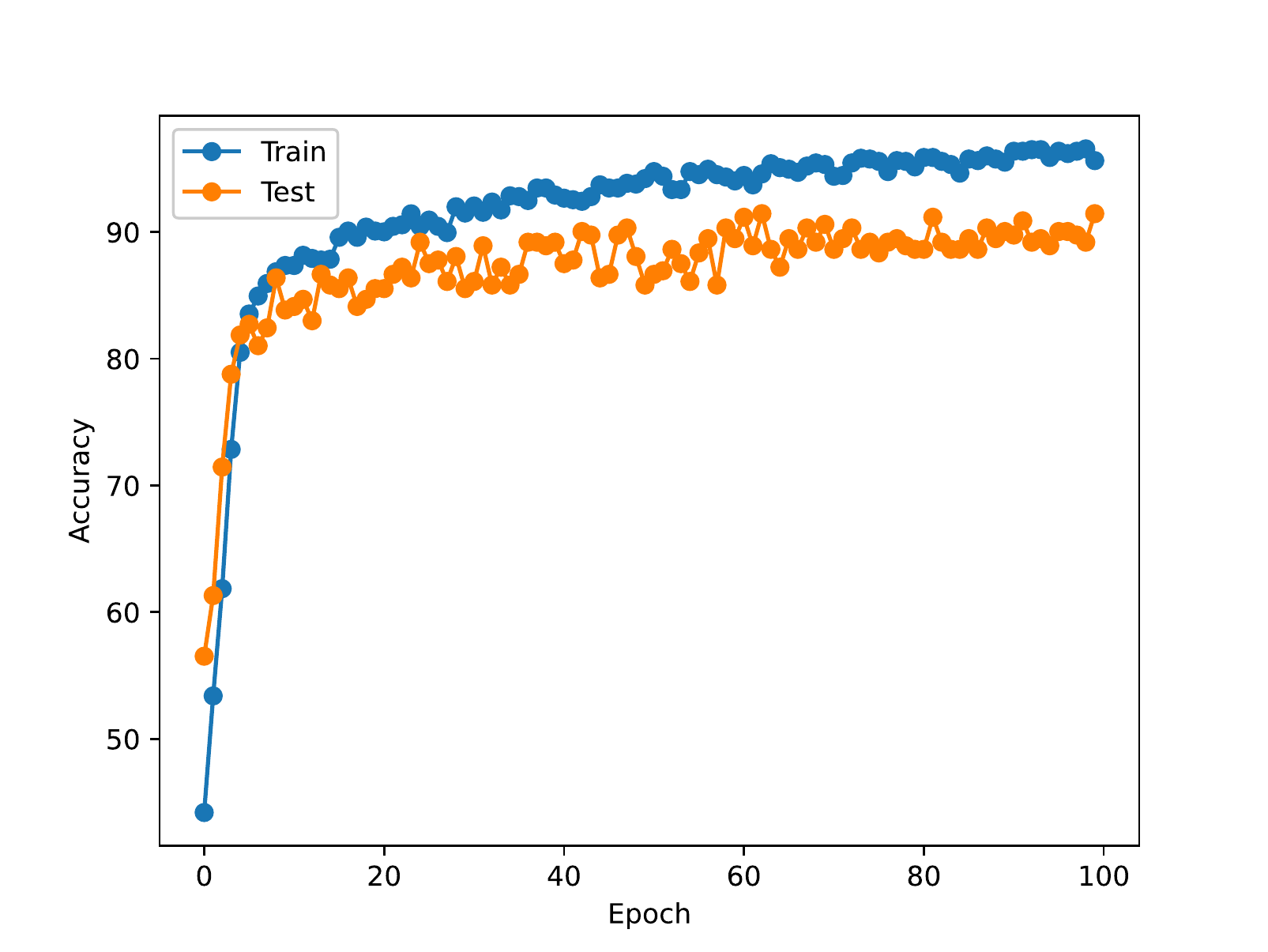}
      \caption{Training and test accuracy for $\text{NN}_1^{\text{I}}$. }
      \label{fig:classify}
    \end{subfigure}
    \hfill
    \begin{subfigure}[b]{0.45\textwidth}
      \centering
      \includegraphics[width=\textwidth]{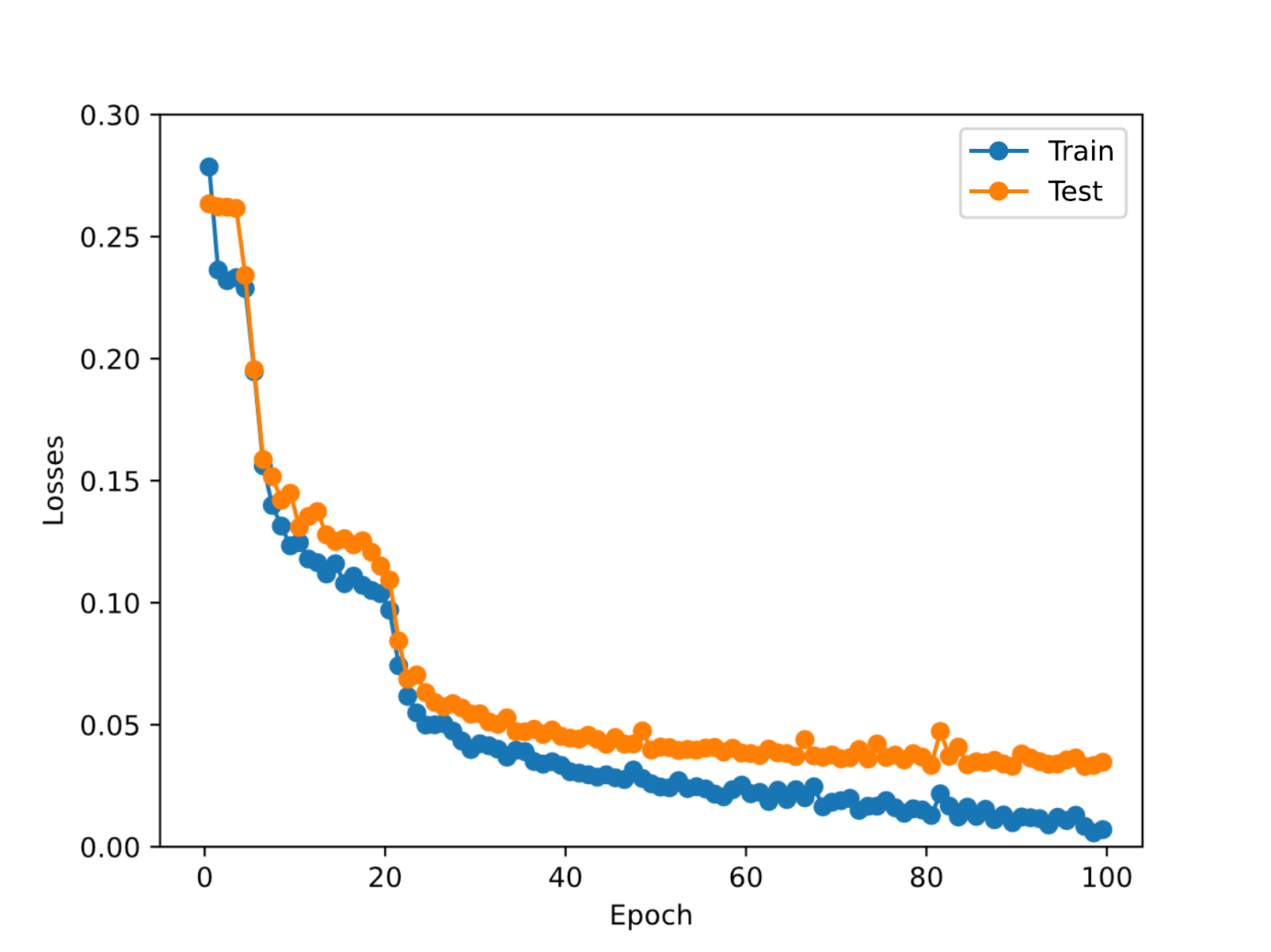}
      \caption{Training and test losses for $\text{NN}_1^{\text{f}}$.}
      \label{fig:regress}
    \end{subfigure}
    \caption{Accuracy and loss of the imitation learning DNNs.}
    \label{fig:imitationlearning}
\end{figure}

\subsection{TOFEL for Autonomous Driving}
\begin{figure}[t!]
    \begin{subfigure}{0.24\textwidth}
      \centering
      \includegraphics[width=\linewidth]{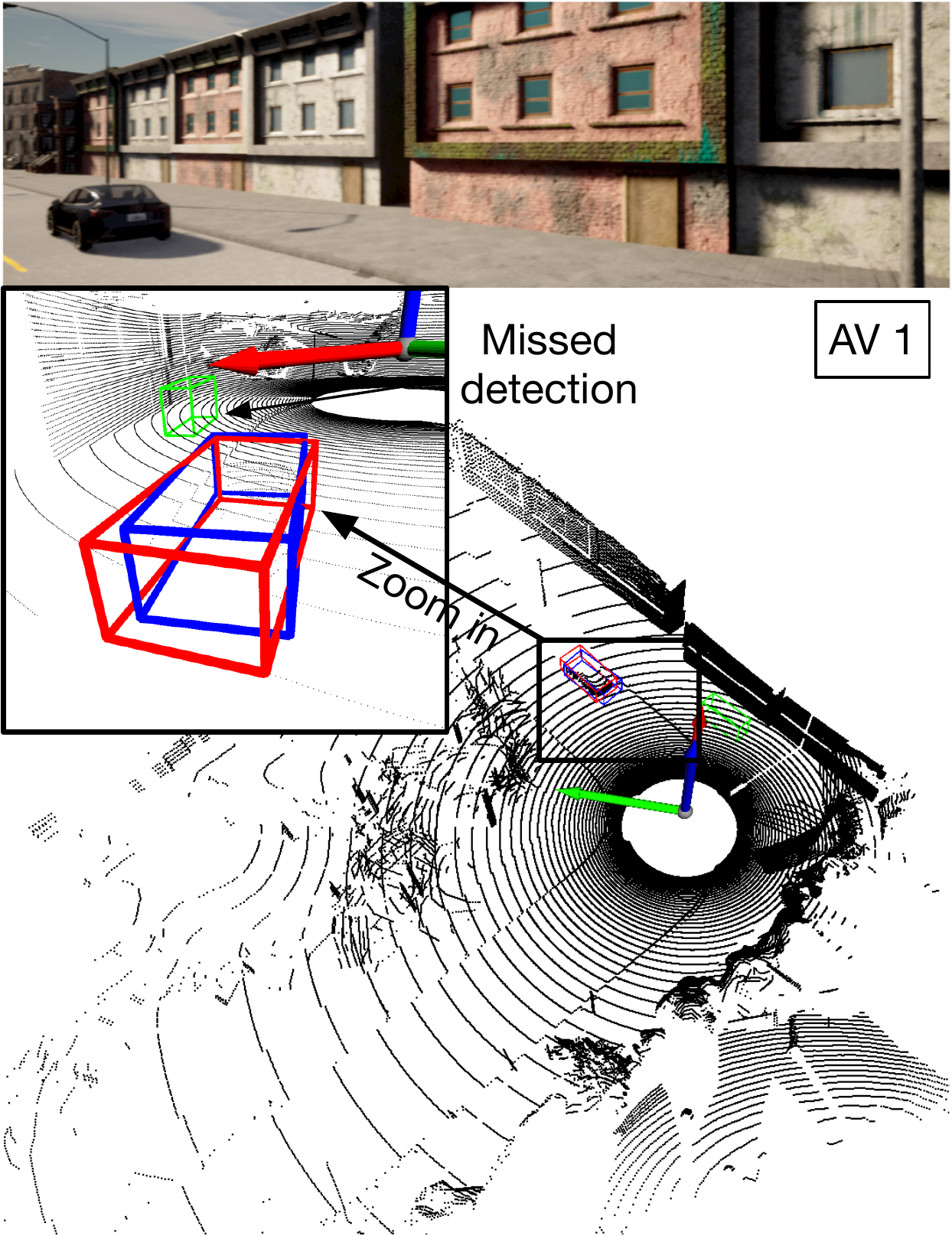}
      \caption{AV 1}
      \label{fig:av1}
    \end{subfigure}
    \begin{subfigure}{0.24\textwidth}
      \centering
      \includegraphics[width=\linewidth]{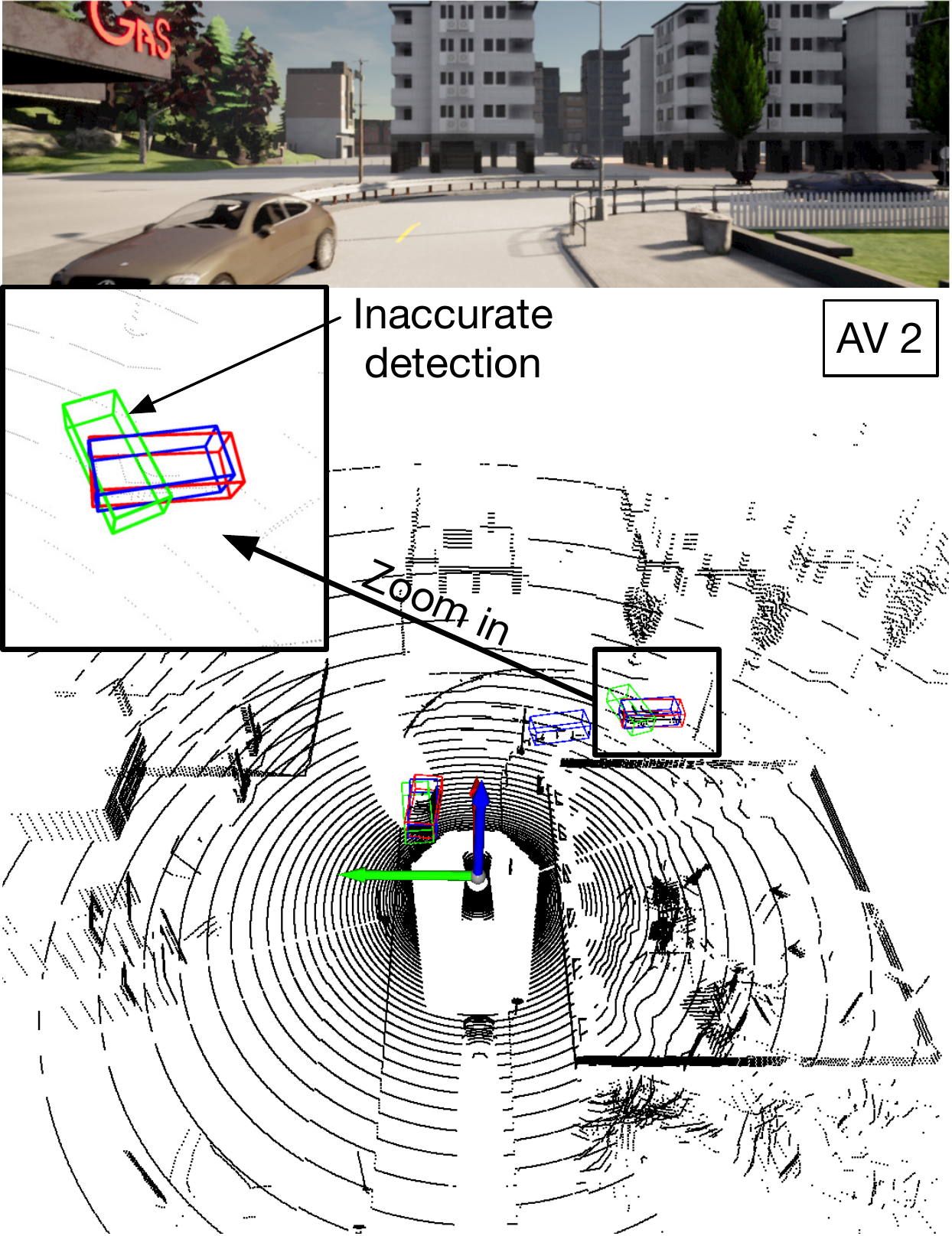}
      \caption{AV 2}
      \label{fig:av2}
    \end{subfigure}
    \begin{subfigure}{0.24\textwidth}
      \centering
      \includegraphics[width=\linewidth]{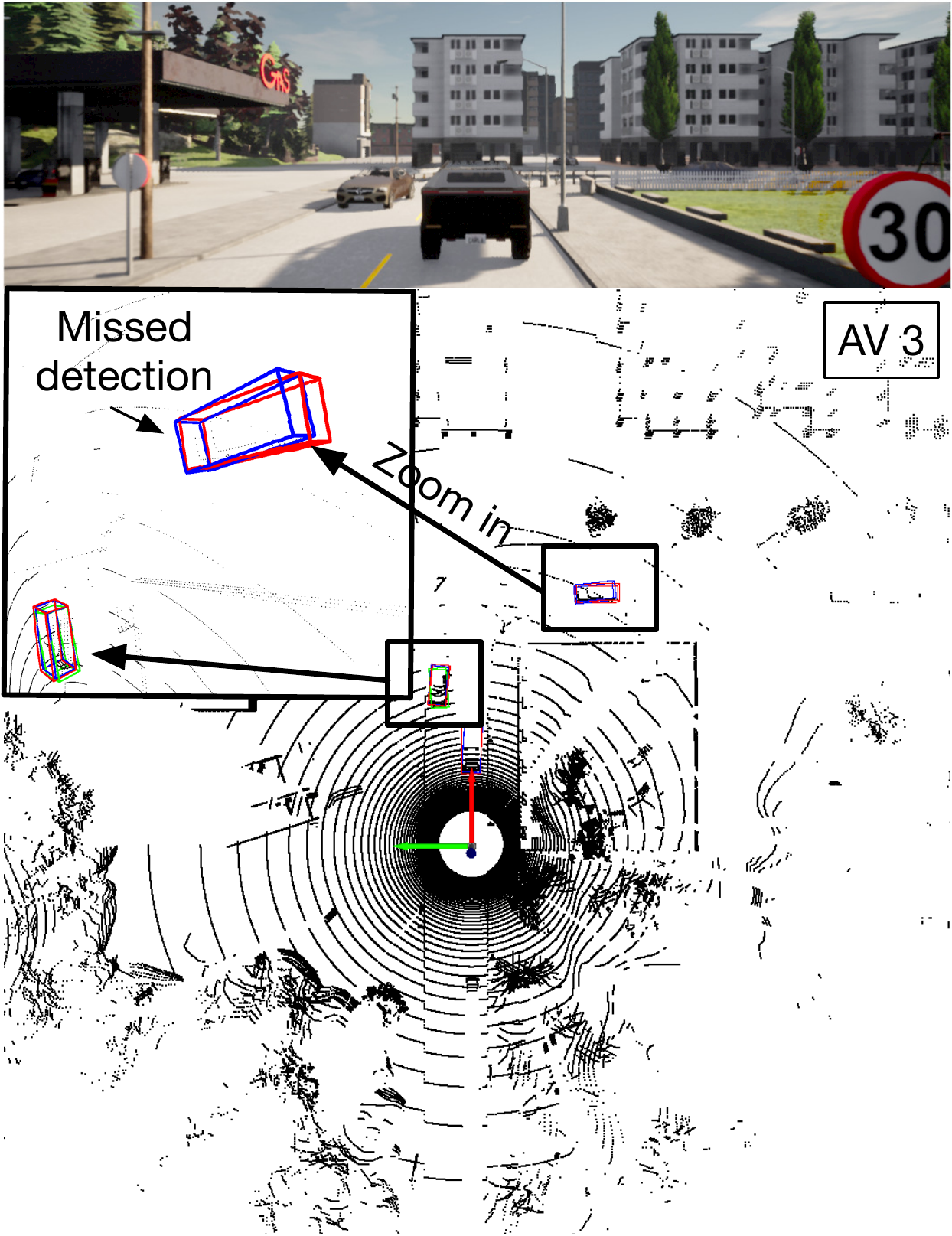}
      \caption{AV 3}
      \label{fig:av3}
    \end{subfigure}
    \begin{subfigure}{0.24\textwidth}
      \centering
      \includegraphics[width=\linewidth]{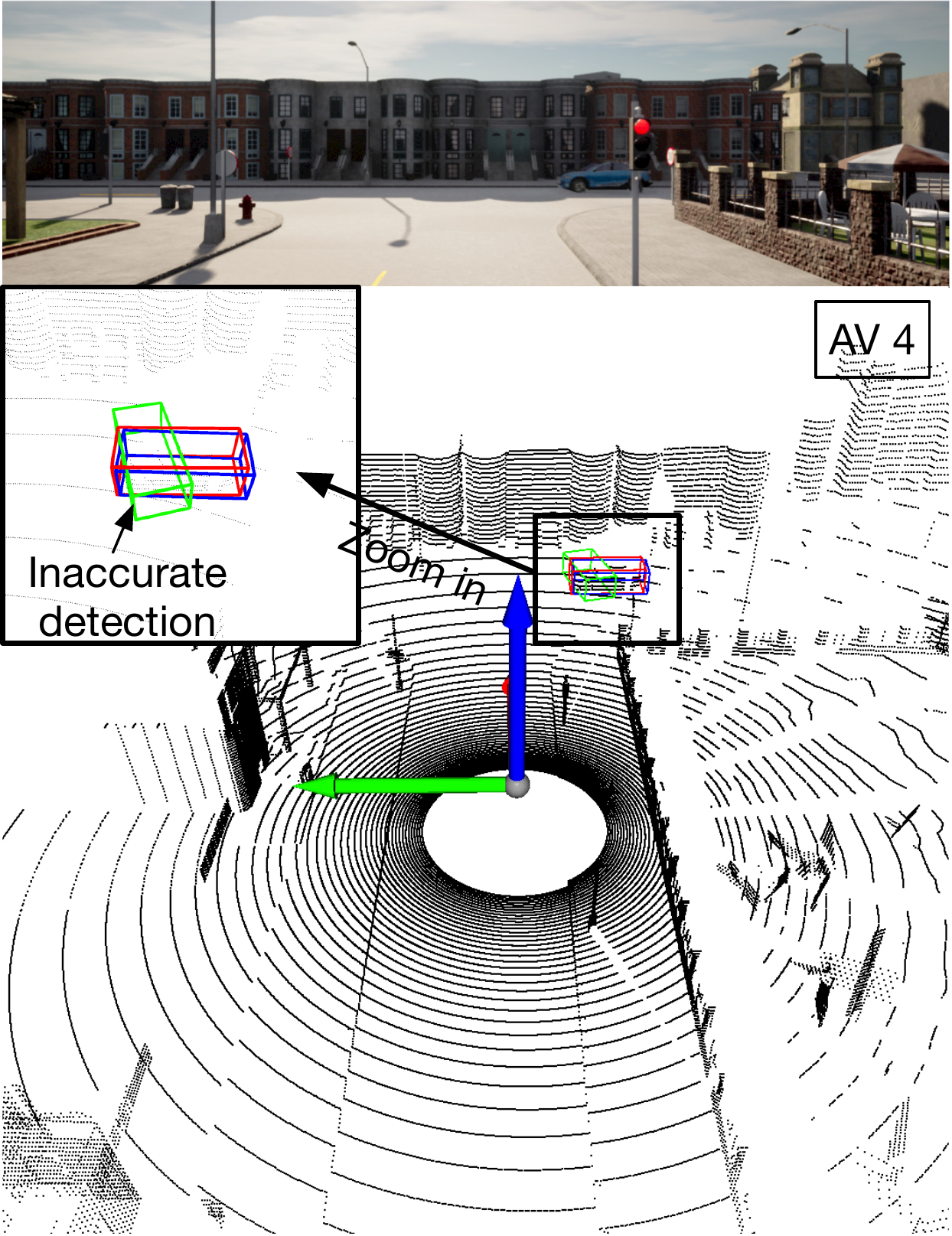}
      \caption{AV 4}
      \label{fig:av4}
    \end{subfigure}
    \caption{Federated detection results of $4$ AVs. The red box represents the ground truth; the blue box is the detection result with TOFEL; the green box is the result of conventional FEEL.}
    \label{fig:avs}
\end{figure}
To verify the effectiveness of our proposed TOFEL framework in more complex learning tasks, we conduct a series of experiments in the scenario of federated learning for 3D object detection in  autonomous vehicle (AV) systems. CARLA \cite{CARLA} is a widely-accepted unreal-engine driven benchmark system that features various urban driving scenarios and state-of-the-art 3D rendering such that TOFEL can be prototyped in virtual-reality. In this paper, all the training and testing procedures are implemented based on CARLA.

\textbf{Dataset}. We use CARLA to generate training and testing data samples. Specifically, we use ``Town02'' map \cite{CARLA} to generate $28$ vehicles, among which $4$ vehicles are AVs that can generate LiDAR point cloud data at the rate of $10\,\mathrm{frames/s}$, and these $4$ AVs perform federated learning to detect the 3D objects in the data frames. In our experiment, each AV generates $600$ frames of data, among which $100$ frames are used for training and the remaining frames are used for testing.
Note that in practice, to obtain the training labels of objects, each vehicle broadcasts its ego position and waits for messages from the nearby infrastructures.
Since the infrastructures are fixed at utility poles and connected to servers via wirelines, they have broader fields of views (FoVs) and deeper neural networks than those of vehicles.
Thus, their outputs are more accurate, which can be transmitted to vehicles via the vehicle-to-infrastructure (V2I) interface and adopted as pseudo labels (i.e., ground-truth labels with noises).

\textbf{DNN Model}. The sparsely embedded convolutional detection (SECOND) neural network \cite{Yan2018SECOND} is adopted for object detection. SECOND is a voxel-based neural network that converts a point cloud to voxel features, and sequentially feeds the voxels into two feature encoding layers, one linear layer, one sparse CNN and one RPN (as detailed in Figure 1 in \cite{Yan2018SECOND}). Notice that the raw data generated from CARLA are not directly compatible with SECOND. To address this issue, we develop a python-based data transformation module, such that the transformed dataset meets the KITTI standard \cite{Geiger2013IJRR,wang2021edge,zhang2021distributed}. The federated learning model training is implemented using PyTorch with python 3.8 on a Linux server with an NVIDIA RTX 3090 GPU.

\textbf{Performance evaluation}. In this part, we compare the performance of TOFEL with conventional FEEL framework. The SECOND network consists of around 5 million parameters. The model size is around $63.7\,\mathrm{Mb}$ as obtained from the experiment. Each frame of the training data is around $1.7\,\mathrm{Mb}$. Moreover, the bandwidth for each vehicle is set to $5\,\mathrm{MHz}$, and the transmit power $P=1\,\mathrm{W}$. The computing speed variation range is $1\sim 10\,\mathrm{GFLOPs}$. Substituting these parameters into problem ($\mathbf {P1}$) and ($\mathbf {P^{\text{flat}}}$), we obtain the latency of one-round iteration for the two schemes as $9.4\,\mathrm{s}$ and $12.1\,\mathrm{s}$, respectively. Given a deadline of $90\,\mathrm{s}$ for federated training, the detection results of the $4$ AVs for our proposed TOFEL scheme and conventional FEEL scheme are shown in Fig. \ref{fig:avs}. It can be observed that almost all the objects can be correctly detected with our proposed TOFEL scheme. In contrast, there are missed detections for AV 1 and AV3 as shown in Fig. \ref{fig:av1} and \ref{fig:av3}, and inaccurate detections (detected with wrong directions) for AV 2 and AV 4 as shown in Fig. \ref{fig:av2} and Fig. \ref{fig:av4}. This results come from the fact that TOFEL can achieve lower latency for one-round FL iteration. Thus, with a given deadline, more global FL iterations can be executed to improve the federated learning accuracy.

\section{Concluding Remarks}
This paper proposed a novel federated learning framework with optimizable topology. The joint design of the topology and computing speed was first solved by a penalty-based SCA method. To facilitate efficient implementation, a deep imitation-learning based framework was proposed to imitate the complex penalty method to achieve real-time decision making. Simulation results validated the effectiveness of our proposed algorithms. Also, it was demonstrated that our proposed TOFEL scheme can remarkably accelerate the federated learning process and reduce the energy consumption. Moreover, the proposed TOFEL scheme was verified in the scenario of federated 3D object detection for V2X autonomous driving.

At a higher level, this paper contributes to the new principle of exploiting hierarchical topology optimization to accelerate federated learning. The existing techniques to boost the efficiency of federated learning, such as gradient/model compression, heterogeneous local update, and power/ bandwidth resource allocation, can be effortlessly built upon our proposed TOFEL framework to further accelerate federated learning.

\appendix
\subsection{Proof of Proposition \ref{prop:equiv}.}
We prove it by mathematical induction. We first consider the case for a 2-tier tree topology. The aggregated gradient at the root node is given by
\begin{align}
    \tilde{\mathbf g}^{(2)}=\frac{\sum_{n=1}^{N_2+1}D_n\mathbf g_n}{\sum_{n=1}^{N_2+1}D_n}=\frac{\sum_{n=1}^{N_2+1}D_n\mathbf g_n}{D^{(2)}},
\end{align}
where $N_2+1$ is the total number of nodes in the system including the root node, and $D^{(2)}$ is the total number of data samples from all nodes in the 2-tier system.
Suppose the following equation satisfies for a $m$-tier tree topology:
\begin{align}
    \tilde{\mathbf g}^{(m)}=\frac{\sum_{n=1}^{N_m+1}D_n\mathbf g_n}{\sum_{n=1}^{N_m+1}D_n}=\frac{\sum_{n=1}^{N_m+1}D_n\mathbf g_n}{D^{(m)}}.
\end{align}
where $N_m+1$ is the total number of nodes in the system including the root node, and $D^{(m)}$ is the total number of data samples from all nodes in the $m$-tier system.
Then, for a system with $m+1$ tiers that consists of a root node and $K$ $m$-tier graphs, we have
\begin{align}
    &\quad\,\tilde{\mathbf g}^{(m+1)}\nonumber\\
    &=\frac{D_1^{(m)}\tilde{\mathbf g}_1^{(m)}+\cdots +D_K^{(m)}\tilde{\mathbf g}_K^{(m)}+D_{N_{m+1}+1}\mathbf g_{N_{m+1}+1}}{D_1^{(m)}+\cdots +D_K^{(m)}+D_{N_{m+1}+1}}\nonumber\\
    &=\frac{\sum_{n=1}^{N_m^1\!+\!1}\!D_n\mathbf g_n\!+\!\cdots \!+\! \sum_{n=1}^{N_m^K+1}\!D_n\mathbf g_n\!+\!D_{N_{m+1}+1}\mathbf g_{N_{m+1}+1}}{D^{(m+1)}}\nonumber\\
    &=\frac{\sum_{n=1}^{N_{m+1}+1}D_n\mathbf g_n}{D^{(m)}}.
\end{align}
Hence, it is established that with an arbitrary tree topology, the final aggregated gradient is equal to that with flat topology.

\subsection{Proof of Proposition \ref{prop:noring}}
We first prove there is no ring in the topology by contradiction. Suppose there is a ring ($N_i\to N_{i+1}\to\cdots\to N_{i+n-1}\to N_i$) with $n$ nodes in the topology. By equation (\ref{eq:CCtimeconst}), we have $\frac{D_{N_i}N_{\mathsf{FLOP}}}{f_{N_i}}\leq \frac{D_{N_{i+1}}N_{\mathsf{FLOP}}}{f_{N_{i+1}}}-\frac{B}{r_{N_i,N_{i+1}}}<\frac{D_{N_{i+1}}N_{\mathsf{FLOP}}}{f_{N_{i+1}}}$. Similarly, $\frac{D_{N_{i+1}}N_{\mathsf{FLOP}}}{f_{N_{i+1}}}<\cdots<\frac{D_{N_{i+n-1}}N_{\mathsf{FLOP}}}{f_{N_{i+n-1}}}$. Thus, we have $\frac{D_{N_i}N_{\mathsf{FLOP}}}{f_{N_i}}<\frac{D_{N_{i+n-1}}N_{\mathsf{FLOP}}}{f_{N_{i+n-1}}}$. However, since node $N_{i+n-1}$ transmits to node $N_i$, we have $\frac{D_{N_{i+n-1}}N_{\mathsf{FLOP}}}{f_{N_{i+n-1}}}\leq \frac{D_{N_i}N_{\mathsf{FLOP}}}{f_{N_i}}-\frac{B}{r_{N_{i+n-1},N_i}}<\frac{D_{N_i}N_{\mathsf{FLOP}}}{f_{N_i}}$. This causes contradiction. Hence, there is no ring in the topology. Moreover, equation (\ref{eq:node_topo_1}) guarantees that any node has and only has one parent node, and equation (\ref{eq:node_topo_2}) ensures the edge server to be the root node. Thus, the whole topology is a tree rooting from the edge server.

\bibliographystyle{IEEEtran}
\bibliography{reference.bib}

\end{document}